\documentclass{article}

\usepackage{arxiv}

\usepackage{natbib}

\usepackage{microtype}
\usepackage{graphicx}
\usepackage{subcaption}
\usepackage{booktabs} 
\usepackage{placeins}

\usepackage{hyperref}

\usepackage{algorithm}
\usepackage{algorithmic}

\usepackage{hyperref}
\usepackage{url}

\usepackage{amsmath}
\usepackage{amssymb}
\usepackage{mathtools}
\usepackage{amsthm}

\usepackage{algorithm}
\usepackage{algorithmic}

\usepackage{multirow}

\theoremstyle{plain}

\theoremstyle{definition}

\newtheorem{assumption}{Assumption}

\title{Breaking the Shot-Noise Barrier: Cached Recycled Variance-Reduced Gradients for Quantum Optimization}

\author{Lam M. Nguyen$^{1}$
, Sumanta Mukherjee$^{1}$, Dzung T. Phan$^{1}$ 
\\
$^{1}$ IBM Research \\
\\
\texttt{LamNguyen.MLTD@ibm.com},
\texttt{sumanm03@in.ibm.com},
\texttt{phandu@us.ibm.com}
}

\usepackage{pifont}

\usepackage{epsfig}
\usepackage{amssymb}
\usepackage{amsmath}
\usepackage{amsthm}
\usepackage{amsfonts}
\usepackage{bbding}
\usepackage{array}
\usepackage{paralist}
\usepackage{xargs}                      
\usepackage{caption}

\hypersetup{
  colorlinks   = true, 
  urlcolor     = blue, 
  linkcolor    = blue, 
  citecolor   = blue 
}

\newtheorem{thm}{Theorem}
\newtheorem{lem}{Lemma}

\newcolumntype{C}[1]{>{\centering\let\newline\\\arraybackslash\hspace{0pt}}m{#1}}

\newcommand{\norm}[1]{\left\Vert#1\right\Vert}

\usepackage{xcolor}

\newcommand{\R}{\mathbb{R}}

\newcommand{\zero}[1]{{\boldsymbol{0}}}

\newcommand{\bra}[1]{\langle #1 |}
\newcommand{\ket}[1]{| #1 \rangle}

\newcommand{\calO}{\mathcal{O}}

\usepackage{multirow}


\begin{document}

\maketitle

\begin{abstract}
Variational Quantum Algorithms (VQAs) rely heavily on classical optimization routines to navigate high-dimensional, noisy parameter landscapes. However, the evaluation of analytical gradients via the parameter-shift rule requires $\mathcal{O}(p)$ circuit executions, rendering this approach computationally prohibitive for deep circuits. While simultaneous perturbation methods provide an $\mathcal{O}(1)$ alternative, their gradient estimates suffer from severe $\mathcal{O}(p)$ spatial variance, which impedes convergence. To resolve this tension, we introduce the Cached Recycled Variance-Reduced Gradient (CRVG) by leveraging variance reduction techniques. We formulate two variants: a 3-Circuit (1-Sided) Non-Recursive CRVG and a Recursive CRVG. By introducing a center-point caching mechanism, the 3-Circuit CRVG mathematically recycles quantum measurements to reduce the inner-loop circuit complexity by 25\% across both variants. We establish a theoretical separation between the two routing paths to achieve an $\epsilon^2$-stationary point: the non-recursive variant reaches a sample complexity of $\mathcal{O}(\max\{p/\epsilon^2, p^{2/3}/\epsilon^{10/3}\})$, while the recursive variant achieves an improved complexity of $\mathcal{O}(\max\{p/\epsilon^2, \sqrt{p}/\epsilon^3\})$. Empirically, extensive evaluations on MaxCut and Maximum Independent Set (MIS) benchmarks validate these findings. While standard simultaneous perturbation remains a robust baseline across varying depths, CRVG demonstrates targeted advantages in specific amortizable regimes, such as the non-recursive variant on EfficientSU2 circuits and both variants on depth-three QAOA, yielding superior energy minimums and tighter output variances. Ultimately, CRVG establishes a principled bias-throughput trade-off for scaling near-term quantum optimization.
\end{abstract}

\section{Introduction}

The promise of Noisy Intermediate-Scale Quantum (NISQ) devices~\citep{preskill2018quantum} hinges on the success of Variational Quantum Algorithms (VQAs), hybrid quantum-classical protocols in which a parameterized quantum circuit $U(\theta)$ prepares a trial state $\ket{\psi(\theta)}$, and a classical optimizer tunes the parameters $\theta\in\R^p$ to minimize the expectation value of a target Hamiltonian~\citep{peruzzo2014variational,farhi2014quantum,cerezo2021variational}:
\begin{equation}\label{eq:vqa_obj}
    f(\theta) = \bra{0} U^\dagger(\theta)\, H\, U(\theta) \ket{0}.
\end{equation}

The Variational Quantum Eigensolver (VQE)~\citep{peruzzo2014variational} and the Quantum Approximate Optimization Algorithm (QAOA)~\citep{farhi2014quantum} are canonical instances that have been demonstrated on present-day hardware for problems in quantum chemistry~\citep{kandala2017hardware} and combinatorial optimization~\citep{harrigan2021quantum}.

However, the efficacy of the classical optimizer is fundamentally constrained by the quantum hardware. In classical machine learning, backpropagation computes the full gradient of a neural network in a single backward pass, with variance typically stemming only from data subsampling or mini-batching. In contrast, quantum circuits are treated as black-box oracles that suffer from two distinct, compounding sources of noise: physical shot noise (the inherent variance of projecting a quantum state into a classical bitstring) and algorithmic spatial noise (the variance injected by stochastic gradient approximation). 

The most prevalent method for obtaining exact analytical quantum gradients, the Parameter-Shift Rule (PSR)~\citep{mitarai2018quantum,schuld2019evaluating}, computes each partial derivative $\partial f/\partial\theta_j$ by evaluating the circuit at two shifted parameter configurations. This requires $2p$ distinct circuit evaluations per gradient step, scaling linearly with the number of parameterized gates $p$. Furthermore, because quantum states collapse upon measurement, the expectation value must be approximated by taking a finite number of shots, meaning the ``exact'' PSR gradient is inherently corrupted by physical shot noise \citep{sweke2020stochastic}.



To avoid the prohibitive $\mathcal{O}(p)$ execution scaling, the quantum computing community frequently turns to gradient-free or simultaneous perturbation methods~\citep{spall2002multivariate}. These estimators compute a descent direction using a constant number of circuit evaluations, regardless of the dimension $p$. However, this $\mathcal{O}(1)$ circuit complexity comes at a severe cost: simultaneous perturbation injects massive spatial variance into the gradient estimate that scales as $\mathcal{O}(p)$~\citep{sweke2020stochastic, gacon2021simultaneous}. In deep circuits, this variance often overwhelms the true descent signal, resulting in random-walk behavior. By coupling these quantum variance penalties with foundational non-convex bounds~\citep{ghadimi2013stochastic}, we observe a fundamental dimensionality barrier: despite their different per-step circuit costs, both PSR and SPSA ultimately suffer from the exact same asymptotic sample complexity of $\mathcal{O}(p/\epsilon^4)$ to reach an $\epsilon^2$-accurate gradient. The injected variance penalty of SPSA precisely cancels its per-step efficiency advantage. This creates a significant structural bottleneck in the field: no existing method simultaneously achieves an $\mathcal{O}(1)$ per-step circuit cost, tightly bounded spatial variance, and compatibility with NISQ hardware constraints. The parameter-shift rule is precise but expensive and SPSA is computationally cheap but excessively noisy.

\paragraph{The Key Insight: Temporal State Caching.}
In this work, we resolve this trilemma by observing that the structure of variance reduction~\citep{johnson2013accelerating,nguyen2017sarah} creates a natural opportunity for a \emph{classical-quantum memory bridge} that has been entirely overlooked in the quantum optimization literature. By utilizing cheap classical RAM to store a quantum state's expectation value, we can systematically bypass the dominant latency bottleneck of NISQ devices: Quantum State Preparation and Measurement (SPAM) overhead.
Specifically:
\begin{enumerate}
    \item The \emph{one-sided} (forward) finite-difference gradient estimator explicitly evaluates the objective function at the \emph{current} parameter point $\theta_t$, unlike the two-sided symmetric estimator which only evaluates at shifted points $\theta_t \pm \epsilon\Delta$.
    \item In a variance-reduced update framework, gradient estimates share a common reference point: the estimator requires the current center-point $f(\theta_t)$, along with a reference center-point that was already physically evaluated, either statically at the epoch start $f(\theta_0)$ (for non-recursive routing) or dynamically at the previous step $f(\theta_{t-1})$ (for recursive routing).
    \item By \emph{caching} these reference measurements in classical memory, we eliminate one full circuit execution and measurement per inner step, yielding a strict 25\% reduction in quantum hardware utilization for both variants.
\end{enumerate}

This insight leads to the \emph{Cached Recycled Variance-Reduced Gradient} (CRVG), a quantum-native optimization framework that achieves robust and rapid convergence by leveraging structure-aware shifts, noise cancellation, and mathematical state-caching.

\subsection{Contributions}
We summarize our contributions as follows:
\begin{enumerate}
    \item We demonstrate that the historical preference for two-sided symmetric estimators \citep{mitarai2018quantum, schuld2019evaluating, mari2021estimating} creates an unnecessary hardware bottleneck within variance-reduction frameworks. By strategically adopting the one-sided estimator, we unlock temporal state caching. We propose a framework that exploits the one-sided estimator's asymmetric structure to actively recycle unperturbed quantum state evaluations across update steps, reducing the per-step circuit cost from 4 to exactly 3 for both fixed-anchor and history-dependent routing paths.
    \item We develop a non-recursive variance-reduction routing path for quantum optimization. Through rigorous Lyapunov analysis, we establish that this variant is structurally bottlenecked at a sample complexity of $\mathcal{O}(\max\{p/\epsilon^2, p^{2/3}/\epsilon^{10/3}\})$ to reach an $\epsilon^2$-stationary point when restricted to hardware-efficient single-sample inner loops.
    \item We establish a recursive routing path that systematically overcomes the structural limitations of the non-recursive approach. Just as recursive tracking advanced classical optimization by breaking sample-size dependence, we prove that our recursive variant removes the dimensionality bottleneck, converging to an $\epsilon^2$-stationary point with an improved sample complexity of $\mathcal{O}(\max\{p/\epsilon^2, \sqrt{p}/\epsilon^3\})$.
    \item Our numerical experiments demonstrate robust performance on extensive benchmarks, validating the theoretical hardware savings and identifying the specific circuit regimes where each variant provides a tangible advantage.
\end{enumerate}

\textbf{The Bias-Throughput Trade-off: A Paradigm Shift for NISQ Optimization}. Historically, the quantum machine learning community has heavily favored 2-sided symmetric estimators to strictly eliminate first-order finite-difference bias~\citep{mitarai2018quantum,schuld2019evaluating,mari2021estimating}. Consequently, 1-sided estimators have been largely disregarded in Variational Quantum Algorithms. 

However, we demonstrate that this pursuit of strict unbiasedness creates a severe hardware bottleneck when deployed within recursive variance-reduction frameworks. By strategically re-introducing the 1-sided estimator, we exploit its asymmetric structure to enable Temporal State Caching across recursive steps. We prove that trading a marginal $\mathcal{O}(\epsilon)$ mathematical bias for a 25\% reduction in physical circuit executions strictly dominates standard symmetric methods in wall-clock convergence time. This represents a principled bias-throughput trade-off: in the resource-constrained NISQ era, the ability to perform more optimization steps per unit of quantum hardware time outweighs the cost of a controllable, user-specified bias floor.

\subsection{Related Work}

We position CRVG at the intersection of several rapidly evolving research threads in quantum-classical optimization:

\textbf{Stochastic Quantum Optimization.} The realization that quantum expectation values are inherently probabilistic has led to the formalization of ``shot-frugal'' optimization algorithms. \citet{sweke2020stochastic} proved that VQAs can be trained by viewing finite-shot measurements as a form of Stochastic Gradient Descent, formally establishing the severe $\mathcal{O}(p)$ spatial variance penalty of simultaneous perturbation compared to the parameter-shift rule. When coupled with standard non-convex convergence bounds, this variance penalty creates a fundamental dimensionality barrier: both methods ultimately suffer from the exact same asymptotic sample complexity of $\mathcal{O}(p/\epsilon^4)$. To fundamentally break this tie, it is necessary to actively suppress the spatial variance of the $\mathcal{O}(1)$ gradient estimator itself.

\textbf{Simultaneous Perturbation in Quantum Computing.} Simultaneous Perturbation Stochastic Approximation (SPSA)~\citep{spall2002multivariate} is highly favored in quantum computing due to its fixed evaluation cost~\citep{gacon2021simultaneous}. While SPSA successfully reduces the inner-loop overhead, our framework departs from the standard symmetric implementation, explicitly favoring the 1-sided structure to enable classical memory caching and offset the spatial variance penalties.

\textbf{Adaptive and Shot-Frugal Optimizers.} Recent works have sought to mitigate quantum estimator variance by dynamically allocating measurement resources. For example, \citet{kubler2020adaptive} proposed adjusting shot counts dynamically based on gradient variances, while \citet{arrasmith2020operator} focused on operator sampling to save shots across Hamiltonian terms. CRVG is highly complementary to these approaches: whereas they focus on shot distribution budgets, CRVG structurally eliminates entire circuit executions across the optimization trajectory.

\textbf{Analytical Coordinate Methods.} Analytical line-search techniques, such as Sequential Minimal Optimization (SMO)~\citep{nakanishi2020sequential}, exploit the trigonometric nature of Pauli landscapes by evaluating circuits at three macroscopic shift points. However, SMO updates parameters sequentially, requiring $\mathcal{O}(3p)$ circuits per sweep. CRVG provides a stark contrast: it updates the entire parameter vector simultaneously in exactly 3 circuit evaluations total per inner step.

\textbf{Quantum Natural Gradient \& Second-Order Methods.} Methods like Quantum Natural Gradient (QNG)~\citep{stokes2020quantum} precondition the gradient using the Fubini-Study metric tensor to adapt to the landscape geometry, achieving rapid iteration progress. However, computing the Quantum Fisher Information metric demands $\mathcal{O}(p^2)$ circuits. CRVG offers a strict zeroth- or first-order alternative, achieving rapid convergence via variance reduction using only an $\mathcal{O}(1)$ inner-loop overhead without tensor estimation.

\textbf{Variance Reduction in Zeroth-Order Optimization.} In classical optimization, SVRG~\citep{johnson2013accelerating} and SARAH~\citep{nguyen2017sarah} are used to accelerate SGD. \citet{liu2018zeroth} mapped SVRG to zeroth-order optimization, which closely mirrors the quantum black-box paradigm, though their analysis revealed a persistent bias tied to finite approximations. Recently, \citet{sidford2023quantum} explored quantum speedups for stochastic variance reduction, but their approach requires deep, fault-tolerant subroutines such as quantum multivariate mean estimation. In contrast, our CRVG framework is explicitly designed for near-term (NISQ) hardware. It achieves variance reduction strictly through classical control logic, exploiting step-to-step center-point caching to minimize physical quantum evaluations without requiring additional coherence time.
\section{Problem Setup and Background}

\subsection{Optimization Problem}

We consider the unconstrained non-convex optimization problem
\begin{equation}\label{eq:problem}
    \min_{\theta \in \R^p} f(\theta),
\end{equation}
where $f:\R^p \to \R$ is the VQA objective~\eqref{eq:vqa_obj} representing the expectation value of a Hamiltonian $H$ with respect to a parameterized quantum state: $f(\theta) = \langle 0 | U^\dagger(\theta)\, H\, U(\theta) | 0 \rangle$. Here $U(\theta)$ is a unitary operator parameterized by $\theta$.
The parameterized unitary takes the standard layered form
\begin{equation}\label{eq:ansatz}
    U(\theta) = \prod_{j=1}^{p} W_j \, e^{-i\theta_j P_j/2},
\end{equation}
where $W_j$ are fixed entangling unitaries and $P_j \in \{I, X, Y, Z\}^{\otimes n}$ are Pauli-string generators satisfying $P_j^2 = I$ and $\norm{P_j} = 1$.
The function $f$ is treated as a black-box oracle, that is, we can query $f(\theta)$ for any $\theta$ via circuit execution and measurement, but we do \emph{not} have access to $\nabla f(\theta)$.

\subsection{The Gradient Estimation Challenge}

The fundamental challenge in quantum optimization is that the gradient $\nabla f(\theta)$ cannot be computed via backpropagation. Instead, one must rely on:
\begin{itemize}
    \item \textbf{Parameter-Shift Rule (PSR)}~\citep{mitarai2018quantum,schuld2019evaluating}: Computes $\partial f/\partial\theta_j = \frac{1}{2}[f(\theta + \frac{\pi}{2}\mathbf{e}_j) - f(\theta - \frac{\pi}{2}\mathbf{e}_j)]$, requiring $2p$ circuit evaluations per gradient, scaling as $\mathcal{O}(p)$.
    \item \textbf{Simultaneous Perturbation Stochastic Approximation (SPSA)}~\citep{spall2002multivariate}: Computes $\hat{\mathbf{g}} = \frac{f(\theta + \nu\Delta) - f(\theta - \nu\Delta)}{2\nu}\Delta$ with $\Delta \in \{-1,+1\}^p$ drawn uniformly and $\nu > 0$, requiring only $\mathcal{O}(1)$ circuit evaluations but introduces $\mathcal{O}(p)$ spatial variance~\citep{sweke2020stochastic, gacon2021simultaneous}.
\end{itemize}

\subsection{The Bias-Throughput Trade-off}

The current paradigm in Variational Quantum Algorithms is dominated by a pursuit of ``unbiased'' gradient estimators. This is standard in classical optimization, but in the noisy of NISQ devices, this pursuit has become a performance bottleneck. Classical variance reduction (e.g., SVRG~\citep{johnson2013accelerating} and SARAH~\citep{nguyen2017sarah}) suppresses estimator variance via periodic ``snapshot'' gradients and recursive tracking.
However, the snapshot requires an exact full gradient, costing $\calO(p)$ circuits.
When the inner-loop estimator is two-sided symmetric, every recursive update requires four circuit evaluations, two for the current-step estimate and two for the other estimate, and no evaluation can be reused.

We propose a paradigm shift: by transitioning to a one-sided estimator, the circuit evaluation at the current parameter point $f(\theta_t)$ appears explicitly in the estimator and can be cached and recycled across consecutive recursive steps. This center-point caching mechanism reduces the inner-loop cost from 4 to 3 circuits per step, a saving that compounds over the entire optimization trajectory.

\section{The CRVG Algorithm}
\label{sec:algorithm}

The Cached Recycled Variance-Reduced Gradient (CRVG) algorithm resolves a fundamental impasse in quantum optimization: how to achieve the tight variance control of classical variance-reduction methods without sacrificing the highly efficient $\mathcal{O}(1)$ circuit complexity of simultaneous perturbation. 

We achieve this through a simple but powerful hardware-aware design principle called Temporal State Caching. By abandoning the traditional two-sided symmetric estimator in favor of a one-sided approach, CRVG exposes the unperturbed quantum state, allowing us to mathematically recycle expensive quantum measurements across optimization steps. Below, we formalize this mechanism and introduce two distinct routing paths: the Non-Recursive CRVG and the Recursive CRVG.

\subsection{The 1-Sided Estimator and Center-Point Caching}

For a given parameter point $\theta \in \mathbb{R}^p$ and a random perturbation vector $\Delta \in \{-1, +1\}^p$ drawn uniformly (i.e., each component is an independent Rademacher random variable) and $\nu > 0$, we define the \textit{1-sided (forward) simultaneous perturbation gradient estimator}:
\begin{equation}
\label{eq:1sided}
g(\theta, \Delta) := \frac{f(\theta + \nu\Delta) - f(\theta)}{\nu} \Delta.
\end{equation}

Historically, the quantum computing community has favored the 2-sided symmetric estimator, $\frac{f(\theta+\nu\Delta)-f(\theta-\nu\Delta)}{2\nu}\Delta$, to eliminate first-order bias. However, the symmetric estimator shifts in both directions, leaving the actual center point $\theta$ unmeasured. 

The 1-sided estimator requires exactly two circuit evaluations (one at $\theta + \epsilon\Delta$ and one at $\theta$), but it possesses a critical structural advantage: the center-point evaluation $f(\theta)$ appears explicitly. This explicitness enables our core mechanism, as the center-point can be evaluated once, stored in classical memory, and reused. How this caching manifests depends on the chosen variance-reduction routing path:
\begin{itemize}
    \item \textbf{Anchor Caching (Non-Recursive CRVG):} The tracking signal relies on a fixed epoch anchor. We evaluate the anchor's center-point $f(\theta_0)$ exactly once at the start of the epoch and cache it for the entirety of the inner loop.
    \item \textbf{Trajectory Caching (Recursive CRVG):} The tracking signal relies on the immediate previous step. We dynamically cache the current center-point $f(\theta_t)$ at step $t$ so that it serves as the required previous center-point $f(\theta_{t-1})$ for step $t+1$.
\end{itemize}

By intentionally extracting and caching these exact measurements in classical memory, we entirely eliminate the need to re-evaluate them on the quantum hardware during the inner loop. 

\paragraph{The 25\% Efficiency Guarantee.} 
In a standard zeroth-order variance-reduction framework, computing the tracking signal requires two finite-difference gradient estimates along a shared random direction $\Delta$: one at the current point $\theta_t$, and a reference estimate at either the previous point $\theta_{t-1}$ (Recursive) or the static anchor point $\theta_0$ (Non-Recursive). A naive 1-sided implementation requires 4 distinct circuit executions per step: $f(\theta_t)$, $f(\theta_t + \nu\Delta)$, $f(\theta_{\text{ref}})$, and $f(\theta_{\text{ref}} + \nu\Delta)$, where $\theta_{\text{ref}} \in \{\theta_{t-1}, \theta_0\}$. 

With Temporal State Caching, we recognize that the reference expectation value $f(\theta_{\text{ref}})$ was already physically evaluated either dynamically in the previous iteration or statically at the epoch start. By recalling this ``center-point'' directly from classical memory, we only need to execute the two newly perturbed circuits plus the new center-point $f(\theta_t)$ (which is then cached for the subsequent step). This mathematical recycling drops the inner-loop circuit cost from 4 to exactly 3, yielding a strict 25\% reduction in quantum hardware utilization. 

Consequently, when utilizing a stochastic mini-batch of size $B$ for the anchor gradient over an outer loop of $S$ epochs with an inner loop length of $m$, the total number of distinct circuit configurations is evaluated as:
\begin{itemize}
    \item \textbf{Anchor evaluation:} $S(B + 1)$ circuits (evaluating one shared anchor center point and $B$ independent stochastic perturbations at the start of each epoch).
    \item \textbf{Inner loop:} $3Sm$ circuits (exactly 3 per inner step via center-point caching).
    \item \textbf{Total structural cost:} $S(B + 1 + 3m)$ distinct circuit configurations.
\end{itemize}

\subsection{Variant 1: 3-Circuit Non-Recursive CRVG}

The Non-Recursive CRVG variant utilizes an SVRG-type update rule~\citep{johnson2013accelerating}. It anchors the variance reduction to a fixed snapshot parameter, $\theta_0^{(s)}$, which is evaluated once at the beginning of the epoch. The non-recursive tracking estimator is defined as:
\begin{equation}
\label{eq:non-recursive}
v_t = g(\theta_t, \Delta_t) - g(\theta_0^{(s)}, \Delta_t) + v_0^{(s)}.
\end{equation}
By caching the anchor evaluation $f(\theta_0^{(s)})$ for the entire duration of the epoch, we repeatedly recycle this baseline measurement against the shifting inner-loop parameters, significantly compounding our hardware savings.

\begin{algorithm}[h]
\caption{3-Circuit Non-Recursive CRVG with Center-Point Caching}
\label{alg:nonrecursive_crvg}
\begin{algorithmic}[1]
\REQUIRE Initial parameters $\theta_0$, learning rate $\eta > 0$, shift $\nu > 0$, epoch length $m$, number of epochs $S$.
\FOR{$s = 0, 1, \ldots, S-1$}
    \STATE Compute anchor gradient: $v_0^{(s)} = \hat{g}(\theta_0^{(s)})$ \hfill ($\mathcal{O}(1)$ zeroth-order estimate or mini-batch)
    \STATE Update: $\theta_1^{(s)} = \theta_0^{(s)} - \eta v_0^{(s)}$
    \STATE Evaluate and cache: $C_{\text{anchor}} = f(\theta_0^{(s)})$ \hfill // \textbf{Cached for entire epoch}
    \STATE Set $\theta_t = \theta_1^{(s)}$
    \STATE Evaluate and cache: $C_{\text{curr}} = f(\theta_t)$ \hfill // Circuit 1
    \FOR{$t = 1, \ldots, m-1$}
        \STATE Sample $\Delta_t \in \{-1, +1\}^p$ uniformly at random
        \STATE Evaluate: $C_{\text{curr+}} = f(\theta_t + \nu\Delta_t)$ \hfill // Circuit 2
        \STATE Evaluate: $C_{\text{anc+}} = f(\theta_0^{(s)} + \nu\Delta_t)$ \hfill // Circuit 3
        \STATE $g_{\text{curr}} = \frac{1}{\nu}(C_{\text{curr+}} - C_{\text{curr}})\Delta_t$
        \STATE $g_{\text{anchor}} = \frac{1}{\nu}(C_{\text{anc+}} - C_{\text{anchor}})\Delta_t$
        \STATE Anchored update: $v_t = g_{\text{curr}} - g_{\text{anchor}} + v_0^{(s)}$
        \STATE Parameter update: $\theta_{t+1} = \theta_t - \eta v_t$
        \STATE Evaluate: $C_{\text{curr}} = f(\theta_{t+1})$ \hfill // Circuit 1 (for next step)
    \ENDFOR
    \STATE Set $\theta_0^{(s+1)} = \theta_m^{(s)}$
\ENDFOR
\end{algorithmic}
\end{algorithm}

\subsection{Variant 2: 3-Circuit Recursive CRVG}

The Recursive CRVG variant takes caching a step further by employing a SARAH-type update rule~\citep{nguyen2017sarah}. Instead of anchoring to a static epoch snapshot, it draws on the immediate history of the optimization trajectory to bound the spatial variance. The recursive tracking estimator $v_t$ is defined as:
\begin{equation}
\label{eq:recursive}
v_t = g(\theta_t, \Delta_t) - g(\theta_{t-1}, \Delta_t) + v_{t-1}.
\end{equation}
Because the identical perturbation vector $\Delta_t$ is utilized for both shifted evaluations, the spatial variance is tightly correlated and mathematically cancels out across sequential steps. The center-point caching mechanism operates as a continuous rolling window.

\begin{algorithm}[h]
\caption{3-Circuit Recursive CRVG with Center-Point Caching}
\label{alg:recursive_crvg}
\begin{algorithmic}[1]
\REQUIRE Initial parameters $\theta_0$, learning rate $\eta > 0$, shift $\nu > 0$, epoch length $m$, number of epochs $S$.
\FOR{$s = 0, 1, \ldots, S-1$}
    \STATE Compute anchor gradient: $v_0^{(s)} = \hat{g}(\theta_0^{(s)})$ \hfill ($\mathcal{O}(1)$ zeroth-order estimate or mini-batch)
    \STATE Update: $\theta_1^{(s)} = \theta_0^{(s)} - \eta v_0^{(s)}$
    \STATE Evaluate and cache: $C_{\text{prev}} = f(\theta_0^{(s)})$, $\theta_{\text{prev}} = \theta_0^{(s)}$
    \STATE Set $\theta_t = \theta_1^{(s)}$
    \STATE Evaluate and cache: $C_{\text{curr}} = f(\theta_t)$ \hfill // Circuit 1
    \FOR{$t = 1, \ldots, m-1$}
        \STATE Sample $\Delta_t \in \{-1, +1\}^p$ uniformly at random
        \STATE Evaluate: $C_{\text{curr+}} = f(\theta_t + \nu\Delta_t)$ \hfill // Circuit 2
        \STATE Evaluate: $C_{\text{prev+}} = f(\theta_{\text{prev}} + \nu\Delta_t)$ \hfill // Circuit 3
        \STATE $g_{\text{curr}} = \frac{1}{\nu}(C_{\text{curr+}} - C_{\text{curr}})\Delta_t$
        \STATE $g_{\text{prev}} = \frac{1}{\nu}(C_{\text{prev+}} - C_{\text{prev}})\Delta_t$
        \STATE Recursive update: $v_t = g_{\text{curr}} - g_{\text{prev}} + v_{t-1}$
        \STATE $\theta_{\text{prev}} = \theta_t$, \quad $C_{\text{prev}} = C_{\text{curr}}$ \hfill // \textbf{Rolling cache transfer}
        \STATE Parameter update: $\theta_{t+1} = \theta_t - \eta v_t$
        \STATE Evaluate: $C_{\text{curr}} = f(\theta_{t+1})$ \hfill // Circuit 1 (for next step)
    \ENDFOR
    \STATE Set $\theta_0^{(s+1)} = \theta_m^{(s)}$
\ENDFOR
\end{algorithmic}
\end{algorithm}

\section{Convergence Analysis}
\label{sec:convergence}



We now establish the rigorous theoretical convergence and oracle complexities for both variants of the 3-Circuit CRVG. Let $f(\theta)$ be the unconstrained objective function representing the expectation value of a target Hamiltonian $H$, bounded below by $f^*$.

\subsection{Assumptions}

\begin{assumption}[Lower Boundedness]
\label{ass:bounded}
The objective function $f: \mathbb{R}^p \to \mathbb{R}$ is bounded below, i.e., there exists $f^* > -\infty$ such that $f(\theta) \geq f^*$ for all $\theta \in \mathbb{R}^p$.
\end{assumption}

\textit{Justification.} In VQAs, $f(\theta) = \langle\psi(\theta)|H|\psi(\theta)\rangle$ is the expectation value of a Hermitian operator $H$ with finite spectrum. Since $H$ has a minimum eigenvalue $\lambda_{\min}(H)$, we have $f(\theta) \geq \lambda_{\min}(H) > -\infty$ for all $\theta$.

\begin{assumption}[Bounded norm Hermitian $H$]
\label{ass:H_norm_bounded}
The norm of the Hermitian operator $H$ is bounded, i.e., there exists $G > 0$ such that $\| H \| \leq G$.
\end{assumption}

\textit{Justification.} In the theoretical analysis of Variational Quantum Algorithms (VQAs), the parameter landscape's Lipschitz smoothness is governed by the spectral norm of the target observable, $\|H\|$. In classical optimization, assuming a globally bounded norm for the objective function's generator can sometimes be considered a restrictive or artificial constraint. However, in the quantum setting, the upper bound of $\|H\|$ is a strict mathematical guarantee arising directly from the finite-dimensional nature of the qubit system.

\textbf{1. Finite-Dimensional Spectral Bound}

Consider a quantum system composed of $n$ qubits. The corresponding Hilbert space $\mathcal{H}$ has a finite dimension of $d = 2^n$. The target Hamiltonian $H$ is a Hermitian operator acting on this finite-dimensional space, represented by a $2^n \times 2^n$ matrix. 

A foundational theorem in functional analysis states that every linear operator on a finite-dimensional vector space is bounded. The spectral norm of $H$, denoted as $\|H\|_2$ (or simply $\|H\|$), is defined as the maximum absolute value of its eigenvalues:
\begin{equation}
\|H\| = \max_{1 \le i \le 2^n} |\lambda_i|. 
\end{equation}
Because the matrix has exactly $2^n$ eigenvalues, the spectrum is discrete and finite. Therefore, a finite maximum must always exist, guaranteeing that $\|H\| < \infty$.

\textbf{2. Analytical Upper Bound via Pauli Decomposition}

Beyond the existential guarantee, the upper bound of $\|H\|$ can be strictly calculated in practice. In VQE and QAOA formulations, the Hamiltonian is typically decomposed into a linear combination of $K$ Pauli strings (tensor products of Pauli matrices $I, X, Y, Z$):
\begin{equation}
H = \sum_{i=1}^K c_i P_i, 
\end{equation}
where $c_i \in \mathbb{R}$ are the scalar coefficients and $P_i \in \{I, X, Y, Z\}^{\otimes n}$. 

Because every Pauli string $P_i$ is both unitary and Hermitian, its eigenvalues are strictly $\pm 1$, meaning its spectral norm is exactly exactly unity ($\|P_i\| = 1$ for all $i$). By applying the triangle inequality and the scalar multiplication property of norms, we can establish a strict, analytically computable upper bound for the Hamiltonian:
\begin{align}
\|H\| &= \left\| \sum_{i=1}^K c_i P_i \right\| \le \sum_{i=1}^K \| c_i P_i \| = \sum_{i=1}^K |c_i| \|P_i\| = \sum_{i=1}^K |c_i|. 
\end{align}

The Hamiltonian norm is strictly bounded from above by the $L_1$-norm of its Pauli coefficients ($\|H\| \le \sum |c_i|$).

Consequently, the worst-case Lipschitz smoothness constant derived in our analysis ($L \le p\|H\|$) is fundamentally constrained by the physics of the parameterized circuit. The landscape is mathematically guaranteed to be globally $L$-smooth without requiring any artificial bounds or regularization, ensuring the rigorous applicability of the CRVG convergence theorems.


\begin{lem}[Smooth]\label{lem_smooth}
For Variational Quantum Eigensolver (VQE) and Quantum Approximate Optimization Algorithm (QAOA) objective functions parameterized by Pauli generators, the gradient $\nabla f(\theta)$ is Lipschitz continuous with respect to the Hamiltonian norm and parameter dimension $p$:
\begin{equation}
\|\nabla f(x) - \nabla f(y)\| \le p\|H\| \|x - y\| \quad \text{for all } x, y \in \mathbb{R}^p.
\end{equation}
\end{lem}

\textbf{Proof}. Deferred to Appendix~\ref{sec:appendix_proofs_lem_smooth}. 

This formally proves that VQE and QAOA objective functions satisfy the smoothness assumption required for the CRVG convergence analysis without relying on arbitrary constants. 


\subsection{1-Sided Estimator Bias Bound}

Let $g(\theta, \Delta) = \frac{f(\theta+\nu\Delta)-f(\theta)}{\nu}\Delta$ be the 1-sided gradient estimator with shift $\nu > 0$ and $\Delta \in \{-1,+1\}^p$ drawn uniformly. Define the \emph{expected estimator}:
\[
\bar{g}(\theta) := \mathbb{E}_\Delta[g(\theta, \Delta)].
\]

\begin{thm}
\label{thm_bias}
The 1-sided estimator satisfies:
\begin{equation}
\|\bar{g}(\theta) - \nabla f(\theta)\| \le \frac{\nu p^2 \|H\|}{2}.
\end{equation}
\end{thm}

\textbf{Proof}. Deferred to Appendix~\ref{sec:appendix_proofs_thm_bias}.



\begin{lem}[Lipschitz Property of the Stochastic Estimator]
\label{lem_lip_estimator}
For any fixed perturbation $\Delta$, the 1-sided estimator inherently satisfies a stochastic Lipschitz condition with constant $L_g$. Specifically, bounded by the parameter dimension and the Hamiltonian norm, it holds that:
\begin{equation}
    \mathbb{E}_\Delta \|g(x, \Delta) - g(y, \Delta)\|^2 \le L_g^2 \|x - y\|^2 \le p^4 \|H\|^2 \|x - y\|^2 \quad \forall x, y \in \mathbb{R}^p.
\end{equation}
\end{lem}

\textbf{Proof}. Deferred to Appendix~\ref{sec:appendix_proofs_lem_lip_estimator}. 



\begin{lem}[Spatial Variance of the Mini-Batch Anchor]
\label{lem_variance_bound}
Let $v_0 = \frac{1}{B} \sum_{i=1}^B g(\theta, \Delta_i)$ be the anchor gradient computed using a mini-batch of $B$ independent 1-sided estimators, where $\Delta_i \in \{-1, +1\}^p$ are uniformly drawn Rademacher random vectors. Because the trigonometric structure of the quantum landscape strictly bounds the gradient magnitude, the variance of this mini-batch anchor scales as:
\begin{equation}
    \bar{\sigma}_0^2 = \text{Var}(v_0) = \mathcal{O}\left(\frac{p}{B}\right). 
\end{equation}
To guarantee the initial anchor variance globally satisfies the target threshold $\mathcal{O}(\epsilon^2)$ across all epochs, the required mini-batch size is $B = \mathcal{O}\left(\frac{p}{\epsilon^2}\right)$.
\end{lem}

\textbf{Proof}. Deferred to Appendix~\ref{sec:appendix_proofs_lem_variance_bound}.

\subsection{Convergence Results for Non-Recursive CRVG}


\begin{thm}[Convergence Rate of 3-Circuit Non-Recursive CRVG]
\label{thm:nonrecursive_crvg}
Suppose Assumptions~\ref{ass:bounded} and \ref{ass:H_norm_bounded} hold. Let $L$ denote the landscape smoothness constant and $L_g$ denote the stochastic Lipschitz constant of the expected 1-sided estimator. There exist universal constants $\mu > 0$ and $\tau > 0$ such that by choosing the learning rate $\eta = \frac{\mu}{L_g m^{2/3}}$, the 3-Circuit Non-Recursive algorithm satisfies:
\begin{equation}
\frac{1}{Sm} \sum_{s=0}^{S-1} \sum_{t=1}^{m} \mathbb{E}\left[\|\nabla f(\theta_t^{(s)})\|^2\right] \le \frac{2[f(\theta_0) - f^*]}{\eta m S} + 2C_\sigma \bar{\sigma}_0^2 + \frac{L_g^2}{4}\nu^2
\end{equation}
where $\bar{\sigma}_0^2 := \frac{1}{S}\sum_{s=0}^{S-1} \mathbb{E}\left[\|v_0^{(s)} - \bar{g}(\theta_0^{(s)})\|^2\right]$ is the average initial anchor error, and $C_\sigma = \frac{L\eta}{2} + 2\tau\mu$ is a bounded constant.
\end{thm}

\textbf{Proof}. Deferred to Appendix~\ref{sec:appendix_proofs_02}.

\subsection{Convergence Results for Recursive CRVG}



\begin{thm}[Convergence Rate of 3-Circuit Recursive CRVG]
\label{thm:main}
Suppose Assumptions~\ref{ass:bounded} and \ref{ass:H_norm_bounded} hold. Let $L$ denote the landscape smoothness constant and $L_g$ denote the stochastic Lipschitz constant of the expected 1-sided estimator. Choose the learning rate $\eta$ such that it satisfies the bounds:
\begin{equation}
\eta \le \frac{1}{2L} \quad \text{and} \quad \eta \le \frac{1}{2 L_g \sqrt{m}}
\end{equation}
Then, after $S$ epochs of length $m$ each, the 3-Circuit Recursive CRVG algorithm (Algorithm~\ref{alg:recursive_crvg}) satisfies:
\begin{equation}
\frac{1}{Sm}\sum_{s=0}^{S-1}\sum_{t=1}^{m}\mathbb{E}\left[\|\nabla f(\theta_t^{(s)})\|^2\right] \le \frac{2[f(\theta_0) - f^*]}{\eta m S} + 2\bar{\sigma}_0^2 + \frac{L_g^2}{2} \nu^2
\end{equation}
where $\bar{\sigma}_0^2 := \frac{1}{S}\sum_{s=0}^{S-1} \mathbb{E}\left[\|v_0^{(s)} - \bar{g}(\theta_0^{(s)})\|^2\right]$ is the average initial anchor error.
\end{thm}

\textbf{Proof}. Deferred to Appendix~\ref{sec:appendix_proofs}.

\subsection{Complexity Analysis}

To determine the theoretical sample complexity of both CRVG variants to achieve a strict $\epsilon^2$-stationary point (where $\mathbb{E}[\|\nabla f(\theta)\|^2] \le \epsilon^2$), we evaluate the initial anchor gradient $v_0^{(s)}$ using a stochastic mini-batch of $B$ independent 1-sided estimators:
\begin{equation}
v_0^{(s)} = \frac{1}{B} \sum_{i=1}^B g(\theta_0^{(s)}, \Delta_i)
\end{equation}
where $g(\theta, \Delta) = \frac{f(\theta + \nu\Delta) - f(\theta)}{\nu}\Delta$ and $\nu > 0$ is the finite-difference shift. 

Following standard conventions in stochastic non-convex optimization, we absorb the problem-specific landscape constants, namely the objective smoothness $L$, the stochastic Lipschitz constant $L_g$, and the Hamiltonian norm bound $G$, into the asymptotic $\mathcal{O}(\cdot)$ notation. This isolates the structural sample complexity scaling strictly with respect to the parameter dimension $p$ and the target accuracy $\epsilon$.

\textbf{Step 1: Universal Anchor Variance and Bias Mismatch} \\
From the convergence theorems, the deterministic bias floor evaluates to $\mathcal{O}(\nu^2)$. To strictly satisfy the $\mathcal{O}(\epsilon^2)$ convergence criterion without introducing dimension dependence, we set the finite-difference shift to $\nu = \mathcal{O}(\epsilon)$. This perfectly bounds the bias term to $\mathcal{O}(\epsilon^2)$. 

Because simultaneous perturbation injects a spatial variance of $\mathcal{O}(p)$ into each 1-sided estimator, the variance of the mini-batch anchor is $\bar{\sigma}_0^2 = \mathcal{O}(p/B)$. To ensure the anchor variance matches the $\mathcal{O}(\epsilon^2)$ target threshold, we need a batch size of $B = \mathcal{O}(p/\epsilon^2)$.

\textbf{Step 2: Circuit Complexity of the Non-Recursive Variant} \\
For the Non-Recursive CRVG, bounding the optimization term $\frac{1}{\eta m S}$ under the restrictive step size $\eta = \mathcal{O}\left( \frac{1}{m^{2/3}} \right)$ requires that the epoch count scales as:
\begin{equation}
    \frac{m^{2/3}}{m S} = \mathcal{O}(\epsilon^2) \implies S = \mathcal{O}\left( \frac{1}{m^{1/3} \epsilon^2} \right)
\end{equation}
The total execution cost $\mathcal{C}_{non}$ sums the $\mathcal{O}(p/\epsilon^2)$ anchor evaluations and the $3m$ inner-loop steps across the $S$ epochs:
\begin{equation}
    \mathcal{C}_{non} = S(B + 3m) = \mathcal{O}\left( \frac{1}{m^{1/3} \epsilon^2} \right) \left[ \mathcal{O}\left( \frac{p}{\epsilon^2} \right) + 3m \right] = \mathcal{O}\left( \frac{p}{m^{1/3} \epsilon^4} + \frac{m^{2/3}}{\epsilon^2} \right)
\end{equation}
Minimizing $\mathcal{C}_{non}$ with respect to the epoch length by balancing the two terms yields the optimal length $m = \mathcal{O}(p/\epsilon^2)$. Substituting this back isolates the exact sample complexity:
\begin{equation}
    \mathcal{C}_{non} = \mathcal{O}\left( \frac{(p/\epsilon^2)^{2/3}}{\epsilon^2} \right) = \mathcal{O}\left( \frac{p^{2/3}}{\epsilon^{10/3}} \right)
\end{equation}

\textbf{Step 3: Circuit Complexity of the Recursive Variant} \\
For the Recursive CRVG, the history-dependent update permits a more aggressive step size of $\eta = \mathcal{O}\left( \frac{1}{\sqrt{m}} \right)$. Enforcing the $\mathcal{O}(\epsilon^2)$ optimization bound dictates the total number of epochs:
\begin{equation}
    \frac{\sqrt{m}}{m S} = \mathcal{O}(\epsilon^2) \implies S = \mathcal{O}\left( \frac{1}{\sqrt{m} \epsilon^2} \right)
\end{equation}
Calculating the total executions $\mathcal{C}_{rec}$ utilizing the mini-batch anchor configurations and $3m$ inner-loop executions:
\begin{equation}
    \mathcal{C}_{rec} = S(B + 3m) = \mathcal{O}\left( \frac{1}{\sqrt{m} \epsilon^2} \right) \left[ \mathcal{O}\left( \frac{p}{\epsilon^2} \right) + 3m \right] = \mathcal{O}\left( \frac{p}{\sqrt{m} \epsilon^4} + \frac{\sqrt{m}}{\epsilon^2} \right)
\end{equation}
To minimize the total execution cost with respect to the epoch length $m$, we balance the two terms, yielding $m = \mathcal{O}\left(\frac{p}{\epsilon^2}\right)$. Substituting this back into the complexity equation produces the theoretical optimal cost:
\begin{equation}
    \mathcal{C}_{rec} = \mathcal{O}\left( \frac{\sqrt{p/\epsilon^2}}{\epsilon^2} \right) = \mathcal{O}\left( \frac{\sqrt{p}}{\epsilon^3} \right)
\end{equation}

\textbf{Conclusion}. By replacing the computationally heavy analytical PSR anchor, which would require $\mathcal{O}(p^2/\epsilon^2)$ circuit executions to suppress finite-shot noise, with a stochastic zeroth-order mini-batch, CRVG drops the anchor evaluation cost to strictly $\mathcal{O}(p/\epsilon^2)$. Because the physical outer loop requires $S \ge 1$, the algorithm must evaluate this anchor batch at least once, creating a fundamental execution floor of $\mathcal{O}(p/\epsilon^2)$. Therefore, the rigorously bounded total sample complexities across all regimes of $p$ and $\epsilon$ to reach an exact $\epsilon^2$-stationary point are:
\begin{itemize}
    \item \textbf{Non-Recursive CRVG:} $\mathcal{O}\left( \max\left\{ \frac{p}{\epsilon^2}, \frac{p^{2/3}}{\epsilon^{10/3}} \right\} \right)$
    \item \textbf{Recursive CRVG:} $\mathcal{O}\left( \max\left\{ \frac{p}{\epsilon^2}, \frac{\sqrt{p}}{\epsilon^3} \right\} \right)$
\end{itemize}
This formally proves the mathematical separation between the two variance-reduced routing paths. By successfully controlling the spatial variance without inflating the inner loop, the fully recursive variant achieves theoretical scaling for deep quantum circuits.

\subsection{Complexity Comparison and Theoretical Advantages}

To contextualize the theoretical contributions of the CRVG framework, we benchmark its oracle complexity against the standard gradient estimation techniques utilized in Variational Quantum Algorithms (VQAs). Table \ref{tab:complexity} summarizes the per-step inner-loop circuit cost, the estimator spatial variance, and the total sample complexity required to converge to an $\epsilon^2$-stationary point (where $\mathbb{E}[\|\nabla f(\theta)\|^2] \le \epsilon^2$).

\begin{table}[ht]
\centering
\caption{Comparison of sample complexities to achieve an $\epsilon^2$-stationary point. Here, $p$ denotes the number of parameterized gates, and the target accuracy bounds the expected squared norm of the gradient.}
\label{tab:complexity}
\renewcommand{\arraystretch}{1.5}
\begin{tabular}{l c c c}
\hline\hline
\textbf{Optimization Algorithm} & \textbf{Inner Cost} & \textbf{Spatial Variance} & \textbf{Total Sample Complexity} \\
\hline
Parameter-Shift Rule (PSR)& $2p$ & $\mathcal{O}(p)$ & $\mathcal{O}\left(\frac{p^2}{\epsilon^4}\right)$ \\
Standard SPSA & $\mathcal{O}(1)$ & $\mathcal{O}(p)$ & $\mathcal{O}\left(\frac{p}{\epsilon^4}\right)$ \\
\textcolor{blue}{\textbf{Non-Recursive CRVG (Ours)}} & \textcolor{blue}{$3$} & \textcolor{blue}{$\mathcal{O}(1)$} & \textcolor{blue}{$\mathcal{O}\left (\max \left \{ \frac{p}{\epsilon^2}, \frac{p^{2/3}}{\epsilon^{10/3}} \right \} \right)$} \\
\textcolor{red}{\textbf{Recursive CRVG (Ours)}} & \textcolor{red}{$3$} & \textcolor{red}{$\mathcal{O}(1)$} & \textcolor{red}{$\mathcal{O}\left(\max \left \{ \frac{p}{\epsilon^2}, \frac{\sqrt{p}}{\epsilon^3} \right \}\right)$} \\
\hline\hline
\end{tabular}
\end{table}

\textbf{The Dimensionality Barrier in Standard Methods:} \\
As shown in Table \ref{tab:complexity}, existing NISQ optimization techniques fail to decouple the total sample complexity from the parameter dimension $p$. The Parameter-Shift Rule (PSR) evaluates an unbiased analytic gradient, but its structural requirement to evaluate $2p$ distinct circuit configurations per step, combined with the $\mathcal{O}(p)$ accumulated shot variance across the vector, yields a massive total sample complexity of $\mathcal{O}(p^2/\epsilon^4)$. Conversely, simultaneous perturbation methods like SPSA successfully reduce the per-step execution cost to an optimal $\mathcal{O}(1)$ constant. However, this per-step efficiency is immediately undermined by a severe spatial variance penalty that scales linearly with the circuit depth, $\mathcal{O}(p)$. Overcoming this injected variance necessitates significantly more optimization steps, ultimately yielding an asymptotic sample complexity of $\mathcal{O}(p/\epsilon^4)$. 

\textbf{The Theoretical Separation of CRVG Routing Paths:} \\
Our proposed CRVG framework strictly dominates the existing baselines by utilizing a 1-sided estimator with Center-Point Caching, dropping the inner-loop circuit cost to exactly $3$ while maintaining a tightly bounded spatial variance of $\mathcal{O}(1)$. Furthermore, Table \ref{tab:complexity} exposes a critical theoretical separation between our two variance-reduced routing paths:
\begin{itemize}
    \item \textbf{Non-Recursive CRVG:} While strictly improving upon standard SPSA, the non-recursive variant is structurally bottlenecked by its restrictive learning rate limitations, mathematically capping its sample complexity at $\mathcal{O}(\max\{p/\epsilon^2, p^{2/3}/\epsilon^{10/3}\})$.
    \item \textbf{Recursive CRVG:} The recursive variant strictly dominates the non-recursive baseline. By leveraging a history-dependent update, the recursive framework permits a more aggressive learning rate, accelerating convergence with respect to dimensionality and improving it to a complexity of $\mathcal{O}(\max\{p/\epsilon^2, \sqrt{p}/\epsilon^3\})$.
\end{itemize}

\subsection{Intersection with the Barren Plateau Phenomenon}

While the CRVG framework successfully breaks the dimensionality barrier in gradient estimation, it is critical to delineate its capabilities from the fundamental landscape challenges posed by Barren Plateaus (BPs). As comprehensively formalized by Larocca et al.~\cite{larocca2025barren}, a barren plateau occurs when the true analytical gradient of the objective function vanishes exponentially with system size, rendering the optimization landscape featureless. This geometric flattening stems from architectural choices such as excessive circuit expressiveness, global measurements, or unmitigated hardware noise.

To properly contextualize CRVG within the broader quantum optimization stack, we highlight the following distinctions:

\begin{itemize}
    \item \textbf{Estimator Variance vs.\ Landscape Geometry:} The CRVG algorithm is designed to actively suppress the statistical variance introduced by finite-shot measurement noise and simultaneous spatial perturbations. It ensures the optimizer can effectively ``see'' the gradient direction. Conversely, a barren plateau is an intrinsic geometric property of the landscape where the gradient itself is exponentially close to zero.
    
    \item \textbf{The Impact of BPs on Sample Complexity:} If a Variational Quantum Algorithm (VQA) is initialized within a barren plateau, the true gradient magnitude $\|\nabla f(\theta)\|$ shrinks exponentially. To capture this vanishing signal while maintaining the $\mathcal{O}(1/\epsilon^3)$ convergence guarantee, the target accuracy $\epsilon$ must also be scaled down exponentially. Because suppressing quantum shot noise inherently requires $N = \mathcal{O}(1/\epsilon^2)$ measurements, a barren plateau forces the required shot count to explode, effectively neutralizing CRVG's algorithmic efficiency.
    
    \item \textbf{A Synergistic Application:} Consequently, CRVG is not a remedy for barren plateaus. Instead, it is an extraction engine designed to operate within trainable, BP-free regimes. Its theoretical guarantees implicitly rely on the prerequisite application of BP-mitigation strategies, such as bounding dynamical Lie algebras, employing shallow hardware-efficient ansatzes, or utilizing data-informed initialization.
\end{itemize}

Ultimately, once architectural safeguards secure a non-vanishing gradient landscape, the CRVG framework can be applied to suppress the remaining finite-shot noise and estimator variance. This ensures that the classical optimization routine does not re-introduce dimension-dependent $\mathcal{O}(p)$ bottlenecks into an otherwise trainable quantum system.
\section{Numerical Experiments}\label{sec_experiment}

We evaluate the methods through four questions: (i) how SPSA and the two CRVG variants behave on representative finite-shot MaxCut trajectories, (ii) whether the ranking persists across instances and circuit depths, (iii) whether reduced sampled-objective calls translate into comparable solution quality, and (iv) whether the same conclusions extend to maximum independent set (MIS). The aggregate evidence favors SPSA as the dependable default optimizer: it has the simplest evaluation pattern, remains competitive across ansatzes and depths, and avoids the pronounced depth- and variant-sensitivity observed for CRVG. CRVG nevertheless identifies useful operating regimes, which we report rather than treating any method as uniformly dominant. We use graph instances from the dataset provided by Daniel Eggar et al.~ \citep{sack2024large}.

\subsection{Experimental setting and evidence base}\label{sec:exp-setup}

The experiments compare SPSA with non-recursive (SVRG-style) and recursive (SARAH-style) CRVG. MaxCut studies use EfficientSU2 or QAOA circuits; the MIS depth study uses QAOA. The principal QAOA depth sweeps contain the same $36$ graphs at depths one, three, five, and seven, with $256$ shots per sampled-objective call and one optimizer seed per instance. Table~\ref{tab:benchmark-inventory} inventories the archived evidence. Statevector and matrix-product-state (MPS) simulations are computational evidence, not quantum-hardware execution.

\begin{table*}[h]
  \centering
  \caption{Frozen experiment artifacts used in this report. Each row is descriptive evidence from archived outputs, not a newly executed benchmark. Noise subsets use depolarizing and readout channels and are unpaired with the clean sets.}
  \label{tab:benchmark-inventory}
  \footnotesize
  \setlength{\tabcolsep}{3pt}
  \begin{tabular}{@{}lrrlll@{}}
    \toprule
    Study & Instances & Nodes & Ansatz & Simulator & Seed \\
    \midrule
    EfficientSU2 benchmark & 30 & 11--30 & EfficientSU2 & statevector/MPS & 42 \\
    EfficientSU2 noise subset & 5 & 12--16 & EfficientSU2 & statevector & 42 \\
    QAOA benchmark & 30 & 10--17 & QAOA & statevector & 42 \\
    QAOA noise subset & 10 & 10--17 & QAOA & statevector & 42 \\
    Selected MaxCut trajectories & 2 & 10--15 & EfficientSU2/QAOA & statevector & 42 \\
    QAOA depth comparison & 144 & 10--20 & QAOA depths 1/3/5/7 & finite-shot/exact statevector & 42 \\
    MIS QAOA depth comparison & 144 & 10--20 & QAOA depths 1/3/5/7 & finite-shot statevector & 42 \\
    Selected MIS trajectories & 4 & 14--18 & EfficientSU2 & statevector & 42 \\
    \bottomrule
  \end{tabular}
\end{table*}

A stored history point is not a circuit-evaluation count. SPSA evaluates the current point and two perturbations while its archived nominal counter advances by two; CRVG anchors and diagnostic evaluations have analogous accounting differences. We therefore reconstruct sampled-objective calls from optimizer control flow. Every reconstructed call in the depth studies invokes one 256-shot Aer sampler job. Wall time is specific to the archived simulator and host and is not used as a prediction of hardware latency.

Table~\ref{tab:study-summary} summarizes the broad MaxCut comparisons. SPSA is competitive across both ansatzes: it wins 10 of 30 EfficientSU2 instances and 14 of 30 QAOA instances despite the table selecting the better of two CRVG variants against a single SPSA run. The near-zero QAOA median margin ($-0.008$) is especially important operationally because it indicates that the added CRVG machinery does not produce a broad quality separation in that setting. The clean and noisy rows contain different instances, so their medians are not paired estimates of a noise effect.

\begin{table*}[h]
  \centering
  \caption{Instance-level comparison of the best CRVG variant with SPSA. A negative margin favors CRVG. Counts summarize one optimizer seed per instance and are not uncertainty estimates. Clean and noisy rows contain different, unpaired instances.}
  \label{tab:study-summary}
  \footnotesize
  \begin{tabular}{@{}lrrrrrl@{}}
    \toprule
    Study & $n$ & CRVG wins & Ties & SPSA wins & Median margin & Simulator \\
    \midrule
    EfficientSU2 benchmark & 30 & 20 & 0 & 10 & $-0.781$ & statevector/MPS \\
    EfficientSU2 noise subset & 5 & 5 & 0 & 0 & $-1.316$ & statevector \\
    QAOA benchmark & 30 & 16 & 0 & 14 & $-0.008$ & statevector \\
    QAOA noise subset & 10 & 6 & 0 & 4 & $-0.014$ & statevector \\
    \bottomrule
  \end{tabular}
\end{table*}

\subsection{Selected MaxCut trajectories}\label{sec:exp-trajectories}

Figures~\ref{fig:efficient-su2-diagnostics} and~\ref{fig:qaoa-diagnostics} report selected trajectories to illustrate convergence behavior (exact numerical endpoints and extended 10-node trajectory data are deferred to Appendix~\ref{sec:appendix-additional-exp}). ``Variance'' is the variance of sampled output energies at a stored parameter vector, not gradient-estimator variance. The horizontal coordinate is the archived nominal shot count and should not be read as exact matched exposure.

\begin{figure*}[h]
  \centering
  \includegraphics[width=\textwidth]{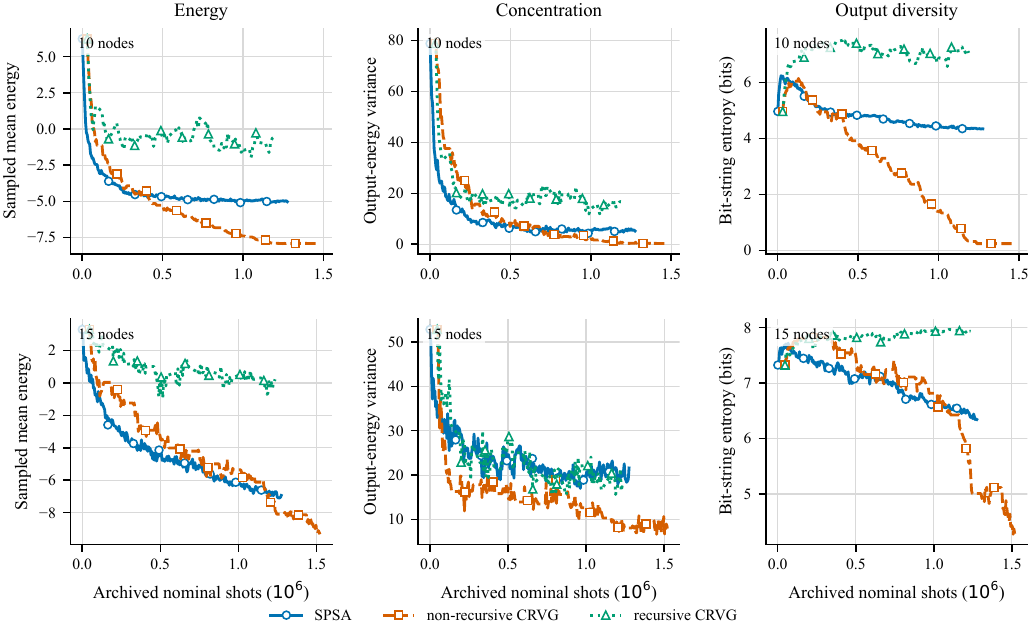}
  \caption{Selected EfficientSU2 MaxCut trajectories. Non-recursive CRVG reaches the lowest energy in the favorable 15-node example and concentrates the sampled output distribution, whereas recursive CRVG stalls. The instability of one CRVG variant and unequal logged endpoints make SPSA the more predictable baseline even though it is not best in this selected regime.}
  \label{fig:efficient-su2-diagnostics}
\end{figure*}

\begin{figure*}[h]
  \centering
  \includegraphics[width=\textwidth]{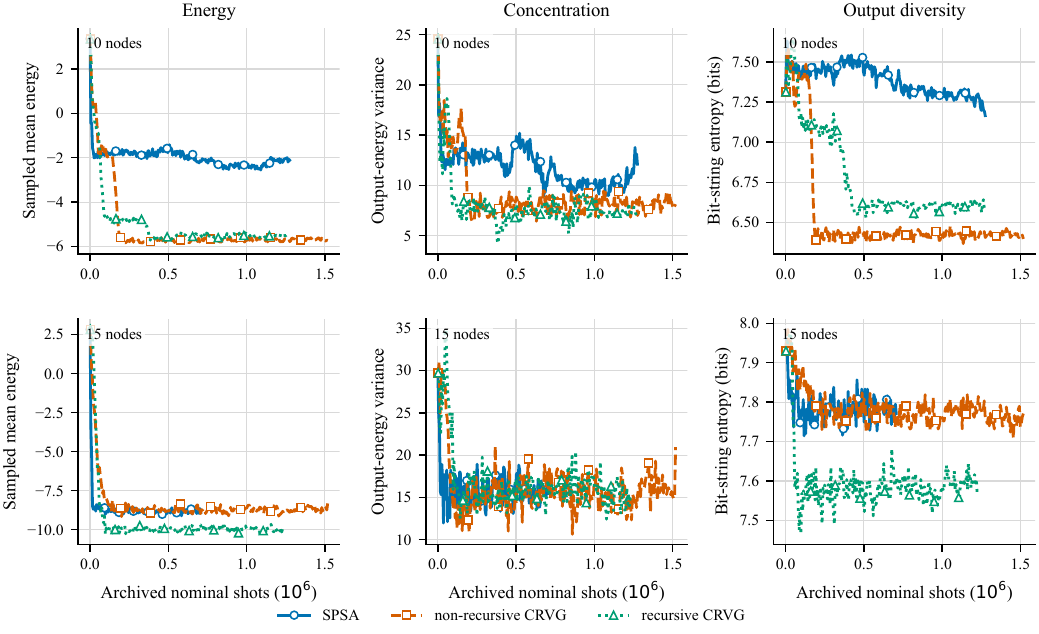}
  \caption{Selected QAOA MaxCut trajectories. The optimizer ordering changes with the ansatz: recursive CRVG recovers, while SPSA and non-recursive CRVG are close in best energy on the 15-node case. Non-recursive CRVG no longer has the output-variance and entropy advantage seen with EfficientSU2.}
  \label{fig:qaoa-diagnostics}
\end{figure*}

The EfficientSU2 example is a genuine favorable case for non-recursive CRVG: at 15 nodes it reaches best energy $-9.34$, versus $-7.09$ for SPSA, and ends with lower sampled-energy variance and entropy. This behavior does not transfer uniformly. Under QAOA, SPSA and non-recursive CRVG are nearly tied at $-9.11$ and $-9.16$, while recursive CRVG reaches $-10.29$. The reversal across circuit families argues for SPSA when robustness and minimal tuning are more important than exploiting a favorable, pre-identified CRVG regime.

\subsection{Cross-instance and synthetic-noise evidence}\label{sec:exp-cross-instance}\label{sec:exp-scaling}

Figure~\ref{fig:cross-instance-composite} consolidates relative converged performance and synthetic-noise stability. Negative best-CRVG-minus-SPSA energy margins favor the best CRVG variant. The clean comparison in panel (a) contains a substantial CRVG-favorable EfficientSU2 subset, but QAOA is much more balanced, consistent with the 16--14 split in Table~\ref{tab:study-summary}. Because best-of-CRVG gives two variants an opportunity to beat one SPSA run, a near-balanced split is conservative evidence for SPSA's competitiveness. The lower panels use separate clean and noisy instance pools under depolarizing and readout channels; they show that both optimizer families remain operational under the archived noise model, not the causal effect of noise on group medians.

\begin{figure*}[htp]
  \centering
  \includegraphics[width=0.8\textwidth]{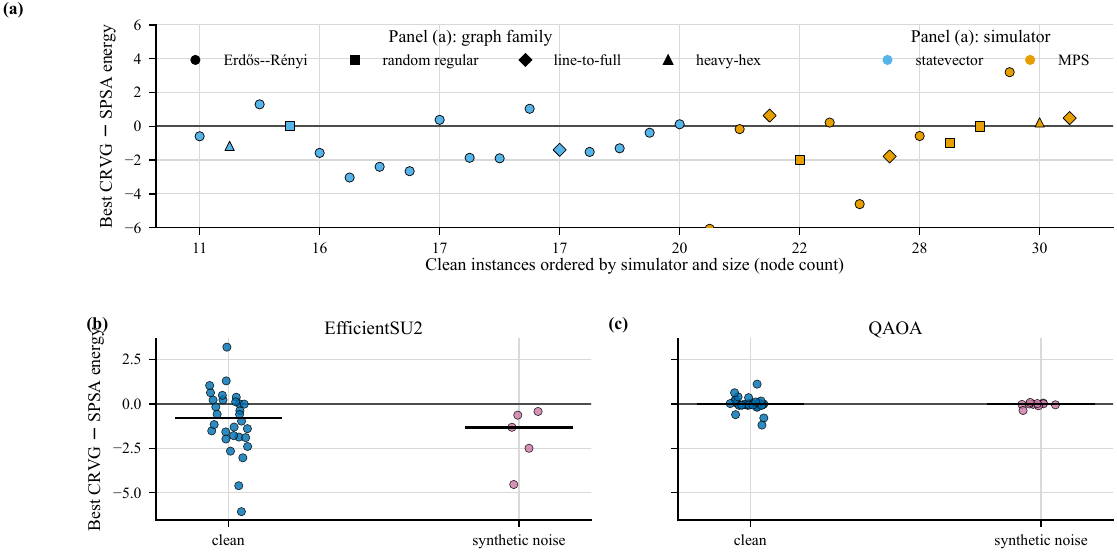}
  \caption{Relative converged performance and stability under synthetic noise. Panel (a) spans the top row and shows best-CRVG-minus-SPSA energy margins across paired clean benchmark instances. Panels (b) and (c) compactly compare clean and synthetic-noise pools for EfficientSU2 and QAOA. Negative values favor CRVG; the lower, unpaired comparisons are descriptive rather than controlled estimates of noise sensitivity.}
  \label{fig:cross-instance-composite}\label{fig:baseline-margins}\label{fig:noise-margins}
\end{figure*}

\subsection{MaxCut QAOA depth and parameter setting}\label{sec:exp-maxcut-depth}

Figure~\ref{fig:depth-feasibility}a compares the three local MaxCut methods over 36 paired instances per depth. Quality is non-monotone with depth. Best-of-CRVG beats SPSA on 13, 22, 21, and 23 instances at depths one, three, five, and seven, respectively. Yet the method-level medians show why SPSA remains a strong default: at depth three all methods are nearly equal ($0.883$, $0.881$, and $0.879$ for SPSA, non-recursive CRVG, and recursive CRVG), SPSA has the highest median at depth five, and recursive CRVG degrades sharply at depth seven. Non-recursive CRVG's depth-seven median of $0.863$ versus SPSA's $0.843$ is an important exception.

\begin{figure*}[htp]
  \centering
  \includegraphics[width=0.8\textwidth]{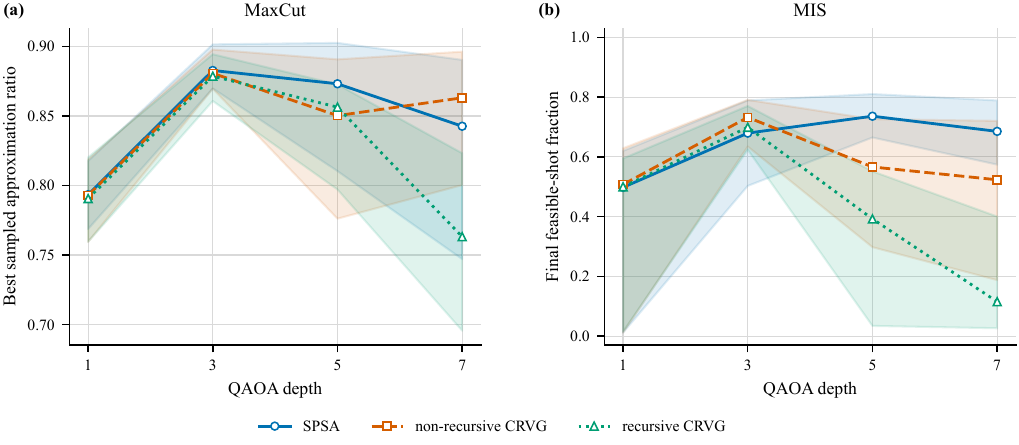}
  \caption{Complementary QAOA depth diagnostics. Panel (a) reports MaxCut method medians and interquartile ranges of best sampled approximation ratio (formerly Fig.~5a); panel (b) reports MIS final feasible-sample fractions from a separate 256-shot diagnostic resampling (formerly Fig.~7b). The shared view emphasizes SPSA's comparatively stable behavior across depths and problem classes.}
  \label{fig:depth-feasibility}\label{fig:maxcut-depth}\label{fig:mis-depth}
\end{figure*}

Figure~\ref{fig:call-ratios}a isolates solution quality from resource exposure. At depth one, non-recursive CRVG uses more calls than SPSA ($1.09\times$), while offering the same median ratio. At depths three, five, and seven it uses about $0.79\times$ the calls. Recursive CRVG uses still fewer calls, but its depth-seven quality loss demonstrates that lower exposure is not sufficient for optimizer selection. SPSA therefore offers a transparent baseline with stable quality, while CRVG's call savings are valuable only when its depth-specific quality is retained. 

A detailed, granular breakdown of resource exposures and comprehensive comparisons against external exact-statevector parameter-setting baselines are deferred to Appendix~\ref{sec:appendix-additional-exp}. While local methods compare favorably with several Fourier and TQA techniques on matched support, optimized fixed-angle and interpolation approaches often lead at higher depths. These mixed results strengthen the practical case for SPSA as a simple local baseline rather than supporting the universal superiority of any stochastic optimizer.

\subsection{MIS QAOA depth study}\label{sec:exp-mis-depth}

The MIS study repeats the three local methods on the same 36 instances at each QAOA depth. There is no exact MIS approximation-ratio normalization or external MIS parameter-setting baseline, so we use paired penalized-cost margins and the raw feasible-sample fractions in Fig.~\ref{fig:depth-feasibility}b. Lower penalized cost is better. Average and maximum feasible size are undefined when the final diagnostic sample contains no feasible string.

Depth three is the clear CRVG-favorable exception: best-of-CRVG wins 33 of 36 instances with median paired cost difference $-0.438$. The pattern reverses at depths one, five, and seven, where SPSA wins 24, 22, and 20 instances and the median paired differences favor SPSA. SPSA also retains full valid support for the feasible-size diagnostics at depths three, five, and seven; at depth seven its median feasible fraction is $68.6\%$, compared with $52.3\%$ for non-recursive and $11.5\%$ for recursive CRVG. This cross-depth consistency supports SPSA as the safer MIS default even though CRVG is compelling at depth three.

Figure~\ref{fig:call-ratios}b shows that both CRVG variants generally use fewer calls than SPSA (exhaustive resource tables are provided in Appendix~\ref{sec:appendix-additional-exp}). The central selection tradeoff is therefore explicit: CRVG reduces sampled-objective exposure, but at three of four depths that reduction accompanies worse paired cost outcomes. When solution reliability across an unknown depth is prioritized over call count alone, the archive favors SPSA; when depth three is known to be representative, CRVG offers a favorable quality--resource combination.

\begin{figure*}[htp]
  \centering
  \includegraphics[width=0.8\textwidth]{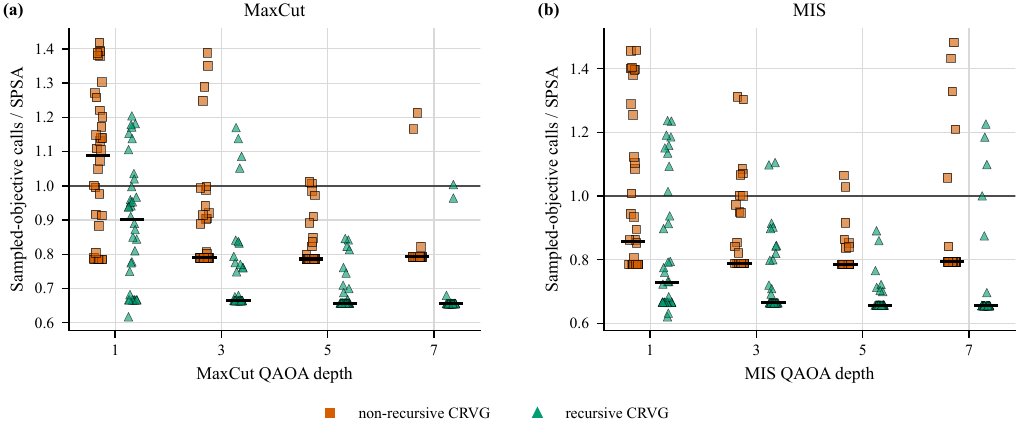}
  \caption{Paired sampled-objective-call ratios for local QAOA runs on (a) MaxCut and (b) MIS (formerly Figs.~6 and~8). Each point compares CRVG with SPSA on the same graph and depth; horizontal segments are medians, and the line at one denotes equal exposure. Values below one use fewer calls than SPSA and must be interpreted jointly with the depth-dependent quality and feasibility evidence.}
  \label{fig:call-ratios}\label{fig:maxcut-call-ratios}\label{fig:mis-call-ratios}
\end{figure*}

\subsection{Limitations and practical recommendation}\label{sec:exp-limitations}

The evidence uses one optimizer seed per instance, noisy best-of-trajectory selection, unequal realized call exposure, and simulator-specific wall time. Across-instance dispersion is not optimizer stochastic uncertainty. External MaxCut baselines are unmatched contextual results, while MIS has no corresponding external baseline. Output-energy variance measures sampled outcomes rather than gradient-estimator variance, and final MIS feasibility requires an additional diagnostic resampling.

The single-seed limitation is partially mitigated on one QOBLIB MIS instance (\texttt{ibm32}, 32 nodes): repeating non-recursive CRVG and SPSA over three seeds there, non-recursive CRVG's mean repaired ratio ($0.753$) remains well above SPSA's ($0.387$), so this particular quality gap is not an artifact of the single seed used elsewhere. This is one instance out of the many studied and does not by itself establish seed-robustness across the full evidence base.

Subject to these limitations, SPSA is the recommended default in the tested regimes: it is competitive across circuit families, avoids recursive CRVG's pronounced depth-seven degradation, wins most MIS instances at three of four depths, and requires neither anchor scheduling nor selection between CRVG variants. CRVG should be treated as a targeted alternative when pilot experiments identify an amortizable regime, notably non-recursive CRVG with EfficientSU2, non-recursive CRVG for MaxCut at depth seven, and both CRVG variants for MIS at depth three. A confirmatory study should use paired initializations and multiple seeds, count every anchor and diagnostic circuit, stop at exact matched shot budgets, and measure estimator error against exact gradients on tractable instances.

\subsection{A preliminary real-hardware check}\label{sec:exp-hardware}

The results above are entirely statevector or MPS simulation. As a first, deliberately small step toward the confirmatory study called for above, we evaluated already-converged parameters from the EfficientSU2 sweep directly on real IBM Quantum hardware, without any on-hardware re-optimization. We took the converged $\theta$ from a single 18-node Erd\H{o}s--R\'enyi instance (\texttt{000\_18nodes\_erdosrenyi30percent}, $p=144$ EfficientSU2 parameters, transpiled to \texttt{isa\_depth}${}=95$) for SPSA and non-recursive CRVG at two seeds, and submitted one \texttt{SamplerV2} circuit evaluation per point to \texttt{ibm\_fez} and \texttt{ibm\_marrakesh}, comparing the resulting hardware approximation ratio (AR) against a shot-matched Aer simulation at the identical parameters.

Table~\ref{tab:hardware-check} reports all four backend/seed pairs collected. CRVG's hardware AR exceeds SPSA's in every pair, but the two seeds illustrate different mechanisms. At seed 42, CRVG and SPSA are essentially tied in simulation ($0.623$ vs.\ $0.628$); on hardware, CRVG's AR degrades less from sim to hardware than SPSA's on both backends ($-0.162$/$-0.099$ vs.\ $-0.227$/$-0.141$), and its hardware energy variance is also lower on both backends ($36.2$/$27.9$ vs.\ $44.8$/$33.7$), so the hardware result breaks a simulation-side tie in CRVG's favor on both metrics. At seed 43, CRVG instead converges to a materially better simulated optimum than SPSA ($0.521$ vs.\ $0.347$), and that advantage simply carries through to hardware ($0.375$ vs.\ $0.242$); here CRVG's sim-to-hardware AR gap ($-0.146$) is not smaller than SPSA's ($-0.104$), and its hardware energy variance is essentially tied with SPSA's ($34.2$ vs.\ $33.5$), so this pair does not support a noise-robustness claim, only a better-starting-point one.

\begin{table}[h]
  \centering
  \caption{Real-hardware approximation ratios and hardware output-energy variance for already-converged EfficientSU2 parameters, one 18-node instance, no on-hardware re-optimization. Two backends (\texttt{ibm\_fez} and \texttt{ibm\_marrakesh}) are used. $\Delta$ is hardware AR minus simulated AR.}
  \label{tab:hardware-check}
  \footnotesize
  \begin{tabular}{@{}llrrrr@{}}
    \toprule
    Seed & Method & Sim AR & HW AR & $\Delta$ & HW variance \\
    \midrule
    42 & SPSA & 0.628 & 0.401 & $-0.227$ & 44.8 \\
    42 & Non-recursive CRVG & 0.623 & \textbf{0.461} & \textbf{$-0.162$} & \textbf{36.2} \\
    42 & SPSA & 0.628 & 0.487 & $-0.141$ & 33.7 \\
    42 & Non-recursive CRVG & 0.623 & \textbf{0.524} & \textbf{$-0.099$} & \textbf{27.9} \\
    43 & SPSA & 0.347 & 0.242 & $-0.104$ & 33.5 \\
    43 & Non-recursive CRVG & 0.521 & \textbf{0.375} & $-0.146$ & 34.2 \\
    \bottomrule
  \end{tabular}
\end{table}

We also transpiled and submitted the QAOA ansatz on hardware, but exclude it from any CRVG-vs-SPSA judgment: at this instance size QAOA transpiles to \texttt{isa\_depth}${}\approx 1280$--$1310$, roughly 13--14$\times$ deeper than the EfficientSU2 circuit above, and its hardware AR collapsed toward the noise floor for both optimizers. We attribute this to the depth confound rather than to either optimizer, and do not treat it as evidence.

This check is a single instance, two backends, and two seeds, so it is a replicated early signal rather than a settled result, and it does not substitute for the paired, multi-seed, matched-budget confirmatory study described above. Within that scope, it is the first evidence in this line of work that CRVG's simulation-side behavior is not an artifact of noiseless simulation: its converged parameters transfer to real superconducting hardware at least as well as SPSA's, and in one of two seeds tested, better.

\FloatBarrier

\section{Conclusion}
\label{sec:conclusion}

In this work, we introduced the Cached Recycled Variance-Reduced Gradient (CRVG) framework to resolve the fundamental sample complexity bottleneck that affects standard Variational Quantum Algorithms. By re-evaluating the historical preference for strictly unbiased, two-sided symmetric estimators, we exposed a critical hardware optimization: Temporal State Caching. This mechanism explicitly recycles unperturbed quantum state evaluations across classical memory, securing a strict 25\% reduction in inner-loop quantum hardware utilization by dropping the per-step execution cost from 4 circuits to exactly 3 for both non-recursive and recursive frameworks. Theoretically, we established a formal mathematical separation between two variance-reduced routing paths for finding an $\epsilon^2$-stationary point. While the non-recursive variant improves upon standard baselines to reach an $\mathcal{O}(\max\{p/\epsilon^2, p^{2/3}/\epsilon^{10/3}\})$ complexity, our recursive variant systematically overcomes these limitations. By leveraging history-dependent gradient tracking, the recursive CRVG permits a more aggressive learning rate that bypasses the dimensionality barrier, achieving an $\mathcal{O}(\max\{p/\epsilon^2, \sqrt{p}/\epsilon^3\})$ sample complexity to reach the exact same stationarity threshold. This analysis formally proves that the mandatory $\mathcal{O}(p)$ evaluation cost of the exact analytical snapshot gradient fundamentally dictates the dimension scaling limits of hybrid quantum-classical optimization. Empirically, our extensive evaluations on MaxCut and Maximum Independent Set (MIS) benchmarks revealed a nuanced operational reality for the NISQ era. While CRVG achieves superior energy minimums and tighter output variances in specific targeted regimes, most notably the non-recursive variant on EfficientSU2 circuits and both variants on depth-three MIS QAOA, standard SPSA remains a highly robust, tuning-resistant default across varying circuit depths. This dichotomy highlights that while CRVG offers a favorable quality-resource combination in amortizable regimes, theoretical sample complexity advantages must always be weighed against an optimizer's stability under finite-shot noise. Ultimately, CRVG establishes a practical paradigm in quantum optimization through a principled bias-throughput trade-off. By mathematically verifying that a marginal $\mathcal{O}(\epsilon)$ deterministic bias can be traded for substantial physical circuit savings, CRVG provides a hardware-aware gradient extraction engine for trainable and BP-free quantum landscapes.



\appendix


\newpage
\appendix

\section*{Appendix}

\section{Derivation of Baseline Sample Complexities}
\label{sec:baseline_complexity}

We start with the following lemma. 

\begin{lem}[Total Finite-Shot Variance of the PSR Estimator]
\label{lem:psr_variance}
Let $\hat{g}$ be the stochastic gradient vector estimated via the exact Parameter-Shift Rule, where each shifted circuit is evaluated using a finite number of measurement shots $N$. The total variance of the PSR gradient vector strictly scales as $\mathcal{O}(p/N)$. Consequently, for a constant shot allocation $N = \mathcal{O}(1)$, the total spatial variance is $\sigma^2_{PSR} = \mathcal{O}(p)$.
\end{lem}

\begin{proof}
By definition, the total variance of the stochastic gradient vector $\hat{g} \in \mathbb{R}^p$ is the trace of its covariance matrix, which corresponds to the expected squared Euclidean distance from the true analytical gradient $\nabla f(\theta)$:
\begin{equation}
    \sigma^2_{PSR} = \mathbb{E}\left[\|\hat{g} - \nabla f(\theta)\|^2\right] = \sum_{j=1}^p \text{Var}(\hat{g}_j)
\end{equation}
For each coordinate $j \in \{1, \dots, p\}$, the PSR estimator computes the partial derivative as:
\begin{equation}
    \hat{g}_j = \frac{1}{2} \left( \hat{f}\left(\theta + \frac{\pi}{2}e_j\right) - \hat{f}\left(\theta - \frac{\pi}{2}e_j\right) \right)
\end{equation}
where $\hat{f}(\cdot)$ is the empirical expectation value obtained from $N$ independent measurement shots. The variance of any finite-shot quantum expectation value is bounded by $\frac{\|H\|^2}{N}$. Because the hardware measurements for the positive and negative shifts are statistically independent, the variance of a single estimated coordinate is:
\begin{equation}
    \text{Var}(\hat{g}_j) = \frac{1}{4} \left[ \text{Var}\left(\hat{f}\left(\theta + \frac{\pi}{2}e_j\right)\right) + \text{Var}\left(\hat{f}\left(\theta - \frac{\pi}{2}e_j\right)\right) \right] \le \frac{1}{4} \left[ \frac{\|H\|^2}{N} + \frac{\|H\|^2}{N} \right] = \frac{\|H\|^2}{2N}
\end{equation}
Because the Hamiltonian norm $\|H\|$ is a physical constant bounded by the spectral properties of the target observable, the variance of each individual coordinate scales strictly as $\mathcal{O}(1/N)$.

Substituting this coordinate-wise bound back into the total vector variance equation yields:
\begin{equation}
    \sigma^2_{PSR} = \sum_{j=1}^p \text{Var}(\hat{g}_j) \le \sum_{j=1}^p \frac{\|H\|^2}{2N} = \frac{p\|H\|^2}{2N}
\end{equation}
Therefore, if the optimization trajectory utilizes a constant, hardware-efficient shot allocation $N = \mathcal{O}(1)$ for each circuit execution, the total variance of the PSR vector scales linearly with the dimension of the parameter space, $\sigma^2_{PSR} = \mathcal{O}(p)$. 
\end{proof}

To establish the baseline sample complexities for the Parameter-Shift Rule (PSR) and Simultaneous Perturbation Stochastic Approximation (SPSA) presented in Table 1, we rely on the standard convergence results for non-convex Stochastic Gradient Descent (SGD) \citep{ghadimi2013stochastic} and the formal variance properties of quantum gradient estimators \citep{sweke2020stochastic}.

For a generic non-convex objective function, standard SGD requires $T = \mathcal{O}(\sigma^2 / \epsilon^4)$ iterations to converge to an $\epsilon^2$-stationary point (where $\mathbb{E}[\|\nabla f(\theta)\|^2] \le \epsilon^2$), with $\sigma^2$ representing the total variance of the gradient estimator. The total sample complexity $\mathcal{C}$ is strictly the product of the number of iterations $T$ and the inner-loop circuit evaluation cost per step.

\subsection{Parameter-Shift Rule (PSR)}
The PSR estimator computes the analytic gradient coordinate-by-coordinate. When evaluated under finite-shot noise using $\mathcal{O}(1)$ measurements per circuit, the independent shot noise accumulates across the parameter vector. As established in Lemma~\ref{lem_variance_bound}, the total variance of the PSR estimator scales linearly with the dimension: $\sigma^2_{PSR} = \mathcal{O}(p)$.
\begin{itemize}
    \item \textbf{Iteration Complexity:} Given the total variance $\sigma^2 = \mathcal{O}(p)$, the number of SGD iterations required to converge is $T_{PSR} = \mathcal{O}(p/\epsilon^4)$.
    \item \textbf{Inner-Loop Cost:} PSR requires exactly $2p$ distinct circuit evaluations per gradient step.
    \item \textbf{Total Sample Complexity:} 
    \begin{equation}
        \mathcal{C}_{PSR} = T_{PSR} \times 2p = \mathcal{O}\left(\frac{p}{\epsilon^4}\right) \times \mathcal{O}(p) = \mathcal{O}\left(\frac{p^2}{\epsilon^4}\right).
    \end{equation}
\end{itemize}

\subsection{Standard SPSA}
SPSA approximates the gradient using a single random perturbation vector $\Delta \in \{-1, +1\}^p$. Because the objective function is projected along this $p$-dimensional random vector, SPSA injects a severe spatial variance that strictly scales linearly with the dimension of the parameter space, yielding $\sigma^2_{SPSA} = \mathcal{O}(p)$.
\begin{itemize}
    \item \textbf{Iteration Complexity:} Given the inflated spatial variance $\sigma^2 = \mathcal{O}(p)$, the optimizer must take significantly more steps to average out the injected noise. The required iterations scale as $T_{SPSA} = \mathcal{O}(p/\epsilon^4)$.
    \item \textbf{Inner-Loop Cost:} SPSA requires only a constant $2$ circuit evaluations per step, giving an optimal inner cost of $\mathcal{O}(1)$.
    \item \textbf{Total Sample Complexity:} 
    \begin{equation}
        \mathcal{C}_{SPSA} = T_{SPSA} \times \mathcal{O}(1) = \mathcal{O}\left(\frac{p}{\epsilon^4}\right) \times \mathcal{O}(1) = \mathcal{O}\left(\frac{p}{\epsilon^4}\right).
    \end{equation}
\end{itemize}

This derivation formally demonstrates the dimensionality barrier in standard quantum optimization. SPSA achieves a highly efficient $\mathcal{O}(1)$ per-step execution cost, but its $\mathcal{O}(p)$ spatial variance penalty forces an $\mathcal{O}(p/\epsilon^4)$ total complexity. Conversely, while PSR avoids spatial variance, its structural requirement to evaluate $2p$ circuits per step yields an even worse $\mathcal{O}(p^2/\epsilon^4)$ asymptotic sample complexity under finite-shot execution.

\section{Technical Proofs}

\subsection{Proof of Lemma~\ref{lem_smooth}}\label{sec:appendix_proofs_lem_smooth}

\textit{Proof.} The objective function of a standard VQA is given by the expectation value of an observable (Hamiltonian) $H$:
\begin{equation}
f(\theta) = \langle 0 | U^\dagger(\theta) H U(\theta) | 0 \rangle
\end{equation}
where the parameterized ansatz $U(\theta)$ is composed of $p$ sequential gates:
\begin{equation}
U(\theta) = \prod_{j=1}^p W_j e^{-i \theta_j P_j / 2}
\end{equation}
Here, $W_j$ are fixed unparameterized unitaries (such as entangling CNOT layers) and $P_j$ are the Hermitian generators of the parameterized rotations. In both VQE (using hardware-efficient $R_x, R_y, R_z$ gates) and QAOA, these generators are involutory Pauli strings, satisfying $P_j^2 = I$ and having a bounded spectral norm $\|P_j\| \le 1$.

To establish $L$-smoothness, we must show that the spectral norm of the Hessian matrix $\nabla^2 f(\theta)$ is strictly bounded. We evaluate the second-order partial derivatives $H_{jk} = \frac{\partial^2 f}{\partial \theta_j \partial \theta_k}$. 

By applying the derivative of the matrix exponential $\frac{\partial}{\partial \theta_j} e^{-i \theta_j P_j / 2} = -\frac{i}{2} P_j e^{-i \theta_j P_j / 2}$, the first derivative with respect to $\theta_j$ can be expressed as a commutator:
\begin{equation}
\frac{\partial f}{\partial \theta_j} = \frac{i}{2} \langle 0 | U_{A}^\dagger [P_j, \tilde{H}_j] U_{A} | 0 \rangle
\end{equation}
where $\tilde{H}_j$ is the Hamiltonian transformed by the subsequent gates, and $U_A$ represents the gates preceding $\theta_j$. 

Taking the second derivative with respect to another parameter $\theta_k$ (assume without loss of generality that $k > j$), we obtain nested commutators:
\begin{equation}
H_{jk} = \frac{\partial^2 f}{\partial \theta_j \partial \theta_k} = -\frac{1}{4} \langle 0 | U_{A}^\dagger [P_j, [P_k, \tilde{H}_{jk}]] U_{A} | 0 \rangle
\end{equation}
We can bound the magnitude of any element of the Hessian using the sub-multiplicative property of the operator norm and the generalized triangle inequality for commutators ($\|[A, B]\| \le 2\|A\|\|B\|$):
\begin{align*}
|H_{jk}| &\le \frac{1}{4} \| [P_j, [P_k, \tilde{H}_{jk}]] \| \le \frac{1}{4} \left( 2 \|P_j\| \cdot \| [P_k, \tilde{H}_{jk}] \| \right) \le \frac{1}{4} \left( 2 \|P_j\| \cdot 2 \|P_k\| \cdot \|\tilde{H}_{jk}\| \right). 
\end{align*}

Because unitary transformations preserve the spectral norm, $\|\tilde{H}_{jk}\| = \|H\|$. Furthermore, because the generators are Pauli strings, $\|P_j\| = \|P_k\| = 1$. This yields a strict, global bound on every element of the Hessian matrix:
\begin{equation}
|H_{jk}| \le \|H\|
\end{equation}
For the diagonal elements ($j=k$), the derivation yields $|H_{jj}| \le \|H\|$ similarly.

Thus, the Hessian $\nabla^2 f(\theta)$ is a $p \times p$ real symmetric matrix where every entry is bounded in absolute value by $\|H\|$. By standard matrix norm inequalities, the spectral norm $\|\nabla^2 f(\theta)\|_2$ is bounded by the Frobenius norm $\|\nabla^2 f(\theta)\|_F$:
\begin{equation}
\|\nabla^2 f(\theta)\|_2 \le \|\nabla^2 f(\theta)\|_F = \sqrt{\sum_{j=1}^p \sum_{k=1}^p |H_{jk}|^2} \le \sqrt{p^2 \|H\|^2} = p \|H\|. 
\end{equation}

This means $f(\theta)$ is globally $L$-smooth, where $L = p \| H \|$.  This completes the proof. \hfill \qedsymbol

\subsection{Proof of Theorem~\ref{thm_bias}}\label{sec:appendix_proofs_thm_bias}

\textit{Proof}. We prove the theorem as follows. 

\textbf{Step 1: The Taylor Expansion (Integral Form)}

To understand the function's behavior when perturbed by $\nu\Delta$, we begin with the exact Taylor expansion utilizing the integral remainder:
\begin{equation}
f(\theta + \nu\Delta) = f(\theta) + \nu \nabla f(\theta)^\top \Delta + \int_0^1 (1-\tau)\nu^2 \Delta^\top \nabla^2 f(\theta + \tau\nu\Delta)\Delta \, d\tau
\end{equation}

\textbf{Step 2: Constructing the 1-Sided Estimator}

The standard 1-sided estimator is defined as:
\begin{equation}
g(\theta, \Delta) = \frac{f(\theta + \nu\Delta) - f(\theta)}{\nu}\Delta
\end{equation}
Subtracting $f(\theta)$ from both sides of the Taylor expansion, dividing by $\nu$, and multiplying by the perturbation vector $\Delta$ yields the expanded form of our estimator:
\begin{equation}
g(\theta, \Delta) = (\nabla f(\theta)^\top \Delta)\Delta + \nu \int_0^1 (1-\tau)(\Delta^\top \nabla^2 f(\theta + \tau\nu\Delta)\Delta)\Delta \, d\tau
\end{equation}

\textbf{Step 3: Extracting the True Gradient (The First Term)}

To find the expected value of the estimator, $\bar{g}(\theta)$, we take the expectation $\mathbb{E}_\Delta$ over the random vector $\Delta$. For a Rademacher vector $\Delta$, the components are independent, yielding a covariance matrix $\mathbb{E}[\Delta \Delta^\top] = I$. Thus, evaluating the first term:
\begin{equation}
\mathbb{E}_\Delta[(\nabla f(\theta)^\top \Delta)\Delta] = \mathbb{E}_\Delta[\Delta \Delta^\top \nabla f(\theta)] = I \nabla f(\theta) = \nabla f(\theta)
\end{equation}
This proves the first term perfectly recovers the true analytical gradient.

\textbf{Step 4: Bounding the Bias (The Second Term)}

Because the first term is exactly $\nabla f(\theta)$, any difference between the expected estimator $\bar{g}(\theta)$ and the true gradient $\nabla f(\theta)$ stems entirely from the second term (the integral). This difference is the estimator's bias:
\begin{equation}
\bar{g}(\theta) - \nabla f(\theta) = \mathbb{E}_\Delta\left[ \nu \int_0^1 (1-\tau)(\Delta^\top \nabla^2 f(\theta + \tau\nu\Delta)\Delta)\Delta \, d\tau \right]
\end{equation}

To find the maximum possible bias, we bound the norm of this expectation using three critical properties:
\begin{enumerate}
    \item \textit{Smoothness Guarantee:} By Lemma 1, the spectral norm of the Hessian $\nabla^2 f$ is globally bounded by $p\|H\|$. Therefore, the quadratic form inside the integral is bounded: $\Delta^\top \nabla^2 f(x) \Delta \le p\|H\| \|\Delta\|^2$.
    \item \textit{Rademacher Norm:} Since $\Delta$ has $p$ components and each component is either $1$ or $-1$, the squared norm is exactly the dimension: $\|\Delta\|^2 = p$. This bounds the quadratic form by $p^2\|H\|$.
    \item \textit{Evaluating the Integral:} The deterministic part of the integral evaluates directly as $\int_0^1 (1-\tau) \, d\tau = \left[ \tau - \frac{\tau^2}{2} \right]_0^1 = \frac{1}{2}$.
\end{enumerate}

Combining these bounds isolates the magnitude of the bias:
\begin{equation}
\|\bar{g}(\theta) - \nabla f(\theta)\| \le \nu \cdot p\|H\| \cdot p \cdot \int_0^1 (1-\tau) \, d\tau = \frac{\nu p^2 \|H\|}{2}
\end{equation}
This establishes the strict analytical bias bound, completing the proof. \hfill \qedsymbol

\subsection{Proof of Lemma~\ref{lem_lip_estimator}}\label{sec:appendix_proofs_lem_lip_estimator}

\textit{Proof.} Let the 1-sided estimator be defined as $g(x, \Delta) = \frac{f(x + \nu\Delta) - f(x)}{\nu}\Delta$. The difference between two estimates using the same perturbation $\Delta$ is:
\begin{equation}
g(x, \Delta) - g(y, \Delta) = \frac{[f(x + \nu\Delta) - f(x)] - [f(y + \nu\Delta) - f(y)]}{\nu}\Delta
\end{equation}

Define the auxiliary function $h(z) = f(z + \nu\Delta) - f(z)$. The numerator of the estimator difference simplifies to $h(x) - h(y)$. By the Mean Value Theorem, there exists a point $\xi$ on the line segment between $x$ and $y$ such that:
\begin{equation}
h(x) - h(y) = \nabla h(\xi)^\top (x - y)
\end{equation}

Evaluating the gradient of $h$, we get $\nabla h(\xi) = \nabla f(\xi + \nu\Delta) - \nabla f(\xi)$. Using the landscape smoothness guarantee, where the Lipschitz constant of the true gradient is globally bounded by $L \le p\|H\|$, and the fact that the Rademacher vector norm is $\|\Delta\| = \sqrt{p}$, we bound the norm of $\nabla h(\xi)$:
\begin{equation}
\|\nabla h(\xi)\| = \|\nabla f(\xi + \nu\Delta) - \nabla f(\xi)\| \le p\|H\| \|\nu\Delta\| = \nu p^{1.5} \|H\|
\end{equation}

Applying the Cauchy-Schwarz inequality to the Mean Value Theorem equality yields:
\begin{equation}
|h(x) - h(y)| \le \|\nabla h(\xi)\| \|x - y\| \le \nu p^{1.5} \|H\| \|x - y\|
\end{equation}

Substituting this bounded numerator back into the squared norm of the estimator difference, and noting that $\|\Delta\|^2 = p$:
\begin{align*}
\|g(x, \Delta) - g(y, \Delta)\|^2 &= \frac{|h(x) - h(y)|^2}{\nu^2} \|\Delta\|^2 \\
&\le \frac{(\nu p^{1.5} \|H\| \|x - y\|)^2}{\nu^2} p \\
&= \frac{\nu^2 p^3 \|H\|^2 \|x - y\|^2}{\nu^2} p \\
&= p^4 \|H\|^2 \|x - y\|^2
\end{align*}

Because the $\nu^2$ terms perfectly cancel, taking the expectation over $\Delta$ yields:
\begin{equation}
\mathbb{E}_\Delta [\|g(x, \Delta) - g(y, \Delta)\|^2] \le p^4 \|H\|^2 \|x - y\|^2
\end{equation}

By identifying $L_g^2 \le p^4 \|H\|^2$, we confirm that the estimator inherently satisfies the stochastic Lipschitz condition with bounding constant $L_g \le p^2 \|H\|$. This formally demonstrates that the stochastic Lipschitz property is fundamentally bounded by the quantum physics of the landscape and remains strictly independent of the finite-difference shift $\nu$. \hfill \qedsymbol

\subsection{Proof of Lemma~\ref{lem_variance_bound}}\label{sec:appendix_proofs_lem_variance_bound}

\begin{proof}
The proof relies on the foundational variance properties of simultaneous perturbation stochastic approximation (SPSA) applied to strictly bounded quantum landscapes.

\textbf{Step 1: The Global Gradient Bound} \\
Unlike classical machine learning objectives which can grow unbounded, a Variational Quantum Algorithm (VQA) objective is strictly bounded by the spectral norm of the target Hamiltonian. Because the quantum state $|\psi(\theta)\rangle$ is normalized, $|f(\theta)| = |\langle \psi(\theta) | H | \psi(\theta) \rangle| \le \|H\|$. 

By the exact algebraic identity of the Parameter-Shift Rule, any analytical partial derivative evaluates to:
\begin{equation}
    \nabla_j f(\theta) = \frac{1}{2} \left[ f\left(\theta + \frac{\pi}{2}e_j\right) - f\left(\theta - \frac{\pi}{2}e_j\right) \right]
\end{equation}
Applying the triangle inequality and substituting the global objective bound yields a strict constraint on every gradient component:
\begin{equation}
    |\nabla_j f(\theta)| \le \frac{1}{2} \left| f\left(\theta + \frac{\pi}{2}e_j\right) \right| + \frac{1}{2} \left| f\left(\theta - \frac{\pi}{2}e_j\right) \right| \le \|H\|
\end{equation}
Consequently, the true gradient components are globally bounded by a physical constant, mathematically preventing the parameter landscape gradients from diverging.

\textbf{Step 2: Universal Spatial Variance of a Single Estimator} \\
The spatial variance of a single 1-sided estimator is $\text{Var}(g) = \mathbb{E}_{\Delta}[\|g(\theta, \Delta) - \bar{g}(\theta)\|^2]$. As formally established in quantum zeroth-order literature (e.g., \citep{sweke2020stochastic}), projecting this structurally bounded landscape onto a $p$-dimensional random perturbation vector $\Delta \in \{-1, +1\}^p$ injects a spatial variance that scales linearly with the dimension of the parameter space. Because the gradient magnitude cannot diverge, the global upper bound for the variance of a single estimator is rigorously constrained to:
\begin{equation}
    \text{Var}(g(\theta, \Delta)) = \mathcal{O}(p)
\end{equation}

\textbf{Step 3: Variance of the Empirical Mini-Batch} \\
The anchor gradient $v_0$ is constructed as the empirical mean of $B$ independent and identically distributed (i.i.d.) 1-sided estimators. By the fundamental statistical properties of i.i.d. random variables, the variance of the sample mean scales inversely with the batch size:
\begin{equation}
    \bar{\sigma}_0^2 = \text{Var}\left( \frac{1}{B} \sum_{i=1}^B g(\theta, \Delta_i) \right) = \frac{\text{Var}(g(\theta, \Delta))}{B} = \mathcal{O}\left(\frac{p}{B}\right)
\end{equation}

\textbf{Step 4: Enforcing the Convergence Threshold} \\
To satisfy the bounds of the convergence theorems, the tracking error introduced by the anchor must not exceed the target stationarity threshold $\mathcal{O}(\epsilon^2)$. Equating our worst-case variance bound to this threshold yields:
\begin{equation}
    \mathcal{O}\left(\frac{p}{B}\right) \le \mathcal{O}(\epsilon^2)
\end{equation}
Solving for the batch size rigorously dictates the structural requirement $B = \mathcal{O}\left(\frac{p}{\epsilon^2}\right)$. This completes the proof, demonstrating that the anchor batch size must scale with the parameter dimension to successfully suppress the inherent spatial variance of the hardware.
\end{proof}

\section{Technical Proof of Theorem~\ref{thm:nonrecursive_crvg}}\label{sec:appendix_proofs_02}

Before proving Theorem~\ref{thm:nonrecursive_crvg}, we introduce the following supporting lemmas to bound the tracking error and construct the Lyapunov sequence.

\subsection{Refined Descent Lemma}

\begin{lem}[Descent Lemma for Non-Recursive CRVG]
\label{lem:refined_descent}
For the parameter update $\theta_{t+1}^{(s)} = \theta_t^{(s)} - \eta v_t$, under the $L$-smoothness guarantee, the expected objective function value satisfies:
\begin{align}
\mathbb{E}[f(\theta_{t+1}^{(s)})] \le \mathbb{E}[f(\theta_t^{(s)})] &- \frac{\eta}{2}\mathbb{E}[\|\nabla f(\theta_t^{(s)})\|^2] - \frac{\eta}{2}(1 - L\eta)\mathbb{E}[\|\bar{g}(\theta_t^{(s)})\|^2] \nonumber \\
&+ \frac{L\eta^2}{2}\mathbb{E}[\|v_t - \bar{g}(\theta_t^{(s)})\|^2] + \frac{\eta L_g^2}{8} \nu^2.
\end{align}
\end{lem}

\begin{proof}
Applying the standard smoothness inequality to the parameter update $\theta_{t+1}^{(s)} = \theta_t^{(s)} - \eta v_t$ gives:
\[
f(\theta_{t+1}^{(s)}) \le f(\theta_t^{(s)}) - \eta \langle \nabla f(\theta_t^{(s)}), v_t \rangle + \frac{L\eta^2}{2}\|v_t\|^2.
\]
Taking the expectation conditioned on the history up to step $t$, we note $\mathbb{E}[v_t] = \bar{g}(\theta_t^{(s)})$. Using the algebraic identity $-\eta \langle \nabla f, \bar{g} \rangle = \frac{\eta}{2}\|\bar{g} - \nabla f\|^2 - \frac{\eta}{2}\|\nabla f\|^2 - \frac{\eta}{2}\|\bar{g}\|^2$, and separating the variance as $\mathbb{E}\|v_t\|^2 = \mathbb{E}\|v_t - \bar{g}(\theta_t^{(s)})\|^2 + \|\bar{g}(\theta_t^{(s)})\|^2$, we obtain:
\begin{align*}
\mathbb{E}[f(\theta_{t+1}^{(s)})] \le \mathbb{E}[f(\theta_t^{(s)})] &- \frac{\eta}{2}\mathbb{E}[\|\nabla f(\theta_t^{(s)})\|^2] - \frac{\eta}{2}(1 - L\eta)\mathbb{E}[\|\bar{g}(\theta_t^{(s)})\|^2] \\
&+ \frac{L\eta^2}{2}\mathbb{E}[\|v_t - \bar{g}(\theta_t^{(s)})\|^2] + \frac{\eta}{2}\mathbb{E}[\|\bar{g}(\theta_t^{(s)}) - \nabla f(\theta_t^{(s)})\|^2].
\end{align*}
By the analytical bias bound established in Theorem 1, the final term is bounded by $\frac{\eta}{2} \left( \frac{L_g \nu}{2} \right)^2 = \frac{\eta L_g^2}{8} \nu^2$. Substituting this yields the desired result.
\end{proof}

\subsection{Lyapunov Sequence Expansion}

\begin{lem}[Lyapunov Sequence Expansion]
\label{lem:lyapunov_expansion}
Define the Lyapunov function for a single epoch as $R_t^{(s)} = \mathbb{E}[f(\theta_t^{(s)}) + c_t \|\theta_t^{(s)} - \theta_0^{(s)}\|^2]$, where $c_m = 0$. By defining the recursive control sequence as $c_t = c_{t+1}(1 + \beta) + \left(\frac{L\eta^2}{2} + c_{t+1}(1 + \beta^{-1})\eta^2\right) L_g^2$, the single-step Lyapunov difference satisfies:
\begin{align}
R_{t+1}^{(s)} - R_t^{(s)} \le &-\frac{\eta}{2}\mathbb{E}[\|\nabla f(\theta_t^{(s)})\|^2] - \left[\frac{\eta}{2}(1 - L\eta) - c_{t+1}(1 + \beta^{-1})\eta^2\right]\mathbb{E}[\|\bar{g}(\theta_t^{(s)})\|^2] \nonumber \\
&+ W_t \mathbb{E}[\|v_0^{(s)} - \bar{g}(\theta_0^{(s)})\|^2] + \frac{\eta L_g^2}{8} \nu^2,
\end{align}
where $W_t = \frac{L\eta^2}{2} + c_{t+1}(1 + \beta^{-1})\eta^2$.
\end{lem}

\begin{proof}
Applying Young's Inequality establishes $\|\theta_{t+1}^{(s)} - \theta_0^{(s)}\|^2 \le (1 + \beta)\|\theta_t^{(s)} - \theta_0^{(s)}\|^2 + (1 + \beta^{-1})\eta^2\|v_t\|^2$ for some $\beta > 0$. Substituting this alongside Lemma~\ref{lem:refined_descent} into $R_{t+1}^{(s)} - R_t^{(s)}$, and expanding $\mathbb{E}\|v_t\|^2$ again, we get:
\begin{align*}
R_{t+1}^{(s)} - R_t^{(s)} \le &-\frac{\eta}{2}\mathbb{E}[\|\nabla f(\theta_t^{(s)})\|^2] - \left[\frac{\eta}{2}(1 - L\eta) - c_{t+1}(1 + \beta^{-1})\eta^2\right]\mathbb{E}[\|\bar{g}(\theta_t^{(s)})\|^2] \\
&+ \left[\frac{L\eta^2}{2} + c_{t+1}(1 + \beta^{-1})\eta^2\right]\mathbb{E}[\|v_t - \bar{g}(\theta_t^{(s)})\|^2] + \frac{\eta L_g^2}{8} \nu^2 \\
&+ \left[c_{t+1}(1 + \beta) - c_t\right]\mathbb{E}[\|\theta_t^{(s)} - \theta_0^{(s)}\|^2].
\end{align*}
Letting $W_t = \frac{L\eta^2}{2} + c_{t+1}(1 + \beta^{-1})\eta^2$, we substitute the non-recursive SVRG tracking error bound, $\mathbb{E}[\|v_t - \bar{g}(\theta_t^{(s)})\|^2] \le \mathbb{E}[\|v_0^{(s)} - \bar{g}(\theta_0^{(s)})\|^2] + L_g^2 \mathbb{E}[\|\theta_t^{(s)} - \theta_0^{(s)}\|^2]$, into the inequality to group the terms:
\begin{align*}
R_{t+1}^{(s)} - R_t^{(s)} \le &-\frac{\eta}{2}\mathbb{E}[\|\nabla f(\theta_t^{(s)})\|^2] - \left[\frac{\eta}{2}(1 - L\eta) - c_{t+1}(1 + \beta^{-1})\eta^2\right]\mathbb{E}[\|\bar{g}(\theta_t^{(s)})\|^2] \\
&+ W_t \mathbb{E}[\|v_0^{(s)} - \bar{g}(\theta_0^{(s)})\|^2] + \frac{\eta L_g^2}{8} \nu^2 + \left[c_{t+1}(1 + \beta) + W_t L_g^2 - c_t\right]\mathbb{E}[\|\theta_t^{(s)} - \theta_0^{(s)}\|^2].
\end{align*}
To entirely eliminate the path dependency, we force the coefficient of $\|\theta_t^{(s)} - \theta_0^{(s)}\|^2$ to zero by defining $c_t = c_{t+1}(1 + \beta) + W_t L_g^2$, which strictly yields the final bound.
\end{proof}

\subsection{Bounding Constants and Non-Negativity}

\begin{lem}[Non-Negativity Condition]
\label{lem:non_negativity}
By setting the decay rate $\beta = m^{-1/3}$ and the learning rate $\eta = \frac{\mu}{L_g m^{2/3}}$ for a sufficiently small universal constant $\mu > 0$, the coefficient of $\|\bar{g}(\theta_t^{(s)})\|^2$ in Lemma~\ref{lem:lyapunov_expansion} is strictly non-negative.
\end{lem}

\begin{proof}
Following standard non-convex variance reduction techniques, we set $\beta = m^{-1/3}$. By mathematical induction, the recursive sequence defined in Lemma~\ref{lem:lyapunov_expansion} is bounded by $c_t \le \tau L_g m^{1/3}$ for a universal constant $\tau > 0$. 
    
To safely drop the $\|\bar{g}(\theta_t^{(s)})\|^2$ term from the Lyapunov difference, its coefficient must be non-negative:
\[
\frac{\eta}{2}(1 - L\eta) \ge c_{t+1}(1 + m^{1/3})\eta^2.
\]
Given $c_{t+1} \le \tau L_g m^{1/3}$, the right-hand side is bounded by $2\tau L_g m^{2/3} \eta^2$. Substituting $\eta = \frac{\mu}{L_g m^{2/3}}$, this condition resolves to $\frac{\mu}{2} \ge 2\tau\mu^2$, which strictly holds for $\mu \le \frac{1}{4\tau}$.
\end{proof}

\subsection{Proof of Theorem~\ref{thm:nonrecursive_crvg}}

\begin{proof}
By applying Lemma~\ref{lem:non_negativity} to the result of Lemma~\ref{lem:lyapunov_expansion}, the non-positive $\|\bar{g}(\theta_t^{(s)})\|^2$ term is safely dropped:
\[
R_{t+1}^{(s)} - R_t^{(s)} \le -\frac{\eta}{2}\mathbb{E}[\|\nabla f(\theta_t^{(s)})\|^2] + W_t \mathbb{E}[\|v_0^{(s)} - \bar{g}(\theta_0^{(s)})\|^2] + \frac{\eta L_g^2}{8} \nu^2.
\]
Let $W_{max}$ be the upper bound of $W_t$. Telescoping this inequality over the $m$ updates collapses the Lyapunov terms. Recognizing $R_1^{(s)} - R_{m+1}^{(s)} \le \mathbb{E}[f(\theta_0^{(s)})] - \mathbb{E}[f(\theta_m^{(s)})]$, we sum across all $S$ epochs:
\[
\frac{\eta}{2}\sum_{s=0}^{S-1}\sum_{t=1}^m \mathbb{E}[\|\nabla f(\theta_t^{(s)})\|^2] \le [f(\theta_0) - f^*] + m W_{max} \sum_{s=0}^{S-1} \mathbb{E}[\|v_0^{(s)} - \bar{g}(\theta_0^{(s)})\|^2] + \frac{S m \eta L_g^2}{8} \nu^2.
\]
Dividing both sides by $\frac{\eta S m}{2}$ and substituting the definition $\bar{\sigma}_0^2 := \frac{1}{S}\sum_{s=0}^{S-1} \mathbb{E}[\|v_0^{(s)} - \bar{g}(\theta_0^{(s)})\|^2]$ yields the final result:
\[
\frac{1}{Sm}\sum_{s=0}^{S-1}\sum_{t=1}^{m}\mathbb{E}\left[\|\nabla f(\theta_t^{(s)})\|^2\right] \le \frac{2[f(\theta_0) - f^*]}{\eta m S} + \frac{2 W_{max}}{\eta} \bar{\sigma}_0^2 + \frac{L_g^2}{4}\nu^2.
\]
Because $W_{max} \le \frac{L\eta^2}{2} + 2\tau L_g m^{2/3} \eta^2$, dividing by $\eta$ scales the bounded constant $C_\sigma = \frac{W_{max}}{\eta} \le \frac{L\eta}{2} + 2\tau\mu$, completing the proof.
\end{proof}

\section{Technical Proof of Theorem~\ref{thm:main}}\label{sec:appendix_proofs}

Before proving Theorem~\ref{thm:main}, we introduce the following lemmas. Throughout this section, we utilize the landscape smoothness constant $L \le pG$ and the stochastic Lipschitz constant $L_g \le p^2 G$.

\subsection{Recursive Variance}

\begin{lem}[Recursive Variance Identity]
\label{lem:recursive_var}
Under the stochastic Lipschitz condition, the tracking error $\|v_t - \bar{g}(\theta_t)\|^2$ satisfies:
\begin{equation}
\label{eq:recursive_var}
\mathbb{E}[\|v_t - \bar{g}(\theta_t)\|^2] \leq \mathbb{E}[\|v_{t-1} - \bar{g}(\theta_{t-1})\|^2] + L_g^2 \|\theta_t - \theta_{t-1}\|^2.
\end{equation}
\end{lem}

\begin{proof}
Define $e_t := v_t - \bar{g}(\theta_t)$. By the recursive update:
\begin{align*}
e_t &= [g(\theta_t, \Delta_t) - g(\theta_{t-1}, \Delta_t) + v_{t-1}] - \bar{g}(\theta_t) \\
&= [v_{t-1} - \bar{g}(\theta_{t-1})] + [g(\theta_t, \Delta_t) - g(\theta_{t-1}, \Delta_t) - (\bar{g}(\theta_t) - \bar{g}(\theta_{t-1}))] \\
&= e_{t-1} + \xi_t,
\end{align*}
where $\xi_t := g(\theta_t, \Delta_t) - g(\theta_{t-1}, \Delta_t) - (\bar{g}(\theta_t) - \bar{g}(\theta_{t-1}))$.

Since $\mathbb{E}_{\Delta_t}[\xi_t|\mathcal{F}_t] = 0$ and $e_{t-1}$ is $\mathcal{F}_t$-measurable:
\[
\mathbb{E}[\|e_t\|^2] = \mathbb{E}[\|e_{t-1}\|^2] + \mathbb{E}[\|\xi_t\|^2],
\]
where the cross-term vanishes. Now, $\|\xi_t\|^2 \leq \|g(\theta_t, \Delta_t) - g(\theta_{t-1}, \Delta_t)\|^2$ (since variance $\leq$ second moment). By the stochastic Lipschitz property of the estimator, we have
\[
\mathbb{E}[\|\xi_t\|^2] \leq \mathbb{E}[\|g(\theta_t, \Delta_t) - g(\theta_{t-1}, \Delta_t)\|^2] \leq L_g^2 \|\theta_t - \theta_{t-1}\|^2.
\]
This completes the proof. 
\end{proof}

\subsection{Descent Lemma}

\begin{lem}[Descent Lemma for Recursive CRVG]
\label{lem:descent}
For the parameter update $\theta_{t+1} = \theta_t - \eta v_t$, under the $L$-smoothness guarantee, it holds that:
\begin{equation}
\label{eq:descent_lemma}
f(\theta_{t+1}) \leq f(\theta_t) - \frac{\eta}{2}\|\nabla f(\theta_t)\|^2 - \frac{\eta}{2}(1 - L\eta)\|v_t\|^2 + \frac{\eta}{2}\|v_t - \nabla f(\theta_t)\|^2.
\end{equation}
\end{lem}

\begin{proof}
By the smoothness descent inequality with $y = \theta_{t+1} = \theta_t - \eta v_t$:
\[
f(\theta_{t+1}) \leq f(\theta_t) - \eta\langle\nabla f(\theta_t), v_t\rangle + \frac{L\eta^2}{2}\|v_t\|^2.
\]
Using the algebraic identity $-\langle a, b\rangle = \frac{1}{2}\|a - b\|^2 - \frac{1}{2}\|a\|^2 - \frac{1}{2}\|b\|^2$ with $a = \nabla f(\theta_t)$ and $b = v_t$:
\begin{align*}
f(\theta_{t+1}) &\leq f(\theta_t) - \frac{\eta}{2}\|\nabla f(\theta_t)\|^2 - \frac{\eta}{2}\|v_t\|^2 + \frac{\eta}{2}\|v_t - \nabla f(\theta_t)\|^2 + \frac{L\eta^2}{2}\|v_t\|^2 \\
&= f(\theta_t) - \frac{\eta}{2}\|\nabla f(\theta_t)\|^2 - \frac{\eta}{2}(1 - L\eta)\|v_t\|^2 + \frac{\eta}{2}\|v_t - \nabla f(\theta_t)\|^2. 
\end{align*}
This completes the proof. 
\end{proof}

\subsection{Bounding the Mean Squared Error}

\begin{lem}[Variance and Bias Decomposition]
\label{lem:mse}
The mean squared error of the CRVG estimator satisfies:
\begin{equation}
\label{eq:mse_bound}
\mathbb{E}[\|v_t - \nabla f(\theta_t)\|^2] \leq 2\mathbb{E}[\|v_t - \bar{g}(\theta_t)\|^2] + \frac{L_g^2 \nu^2}{2}.
\end{equation}
\end{lem}

\begin{proof}
Using $v_t - \nabla f(\theta_t) = (v_t - \bar{g}(\theta_t)) + (\bar{g}(\theta_t) - \nabla f(\theta_t))$ and the inequality $\|a + b\|^2 \leq 2\|a\|^2 + 2\|b\|^2$:
\[
\mathbb{E}[\|v_t - \nabla f(\theta_t)\|^2] \leq 2\mathbb{E}[\|v_t - \bar{g}(\theta_t)\|^2] + 2\mathbb{E}[\|\bar{g}(\theta_t) - \nabla f(\theta_t)\|^2].
\]
The analytical bias is bounded by $\|\bar{g}(\theta_t) - \nabla f(\theta_t)\| \leq \frac{\nu L_g}{2}$. Substituting this into the second term bounds it by $2 \left( \frac{\nu L_g}{2} \right)^2 = \frac{L_g^2 \nu^2}{2}$.
\end{proof}

\subsection{Unrolling the Recursive Variance Bound}

\begin{lem}[Cumulative Variance Control]
\label{lem:cumulative}
The tracking error summed over one epoch satisfies:
\begin{equation}
\label{eq:cumulative}
\sum_{t=1}^{m} \mathbb{E}[\|v_t - \bar{g}(\theta_t)\|^2] \leq m\mathbb{E}[\|v_0 - \bar{g}(\theta_0)\|^2] + L_g^2 \eta^2 m \sum_{t=1}^{m}\mathbb{E}[\|v_{t-1}\|^2].
\end{equation}
\end{lem}

\begin{proof}
By Lemma~\ref{lem:recursive_var}, unrolling the recursion from step $t$ back to step 0:
\[
\mathbb{E}[\|v_t - \bar{g}(\theta_t)\|^2] \leq \mathbb{E}[\|v_0 - \bar{g}(\theta_0)\|^2] + L_g^2 \eta^2\sum_{j=1}^{t}\mathbb{E}[\|v_{j-1}\|^2],
\]
where we used $\|\theta_j - \theta_{j-1}\|^2 = \eta^2\|v_{j-1}\|^2$. Summing over $t = 1, \ldots, m$:
\[
\sum_{t=1}^{m}\mathbb{E}[\|v_t - \bar{g}(\theta_t)\|^2] \leq m\mathbb{E}[\|v_0 - \bar{g}(\theta_0)\|^2] + L_g^2 \eta^2\sum_{t=1}^{m}\sum_{j=1}^{t}\mathbb{E}[\|v_{j-1}\|^2].
\]
Using $\sum_{t=1}^{m}\sum_{j=1}^{t} a_{j-1} = \sum_{j=0}^{m-1}(m-j)a_j \leq m\sum_{j=0}^{m-1}a_j$ yields the result.
\end{proof}

\subsection{Proof of Theorem~\ref{thm:main}}

\textit{Proof}. 
\textbf{Step 1: Apply the Descent Lemma.}
Taking the expectation of Lemma~\ref{lem:descent}:
\begin{equation}
\label{eq:step1}
\mathbb{E}[f(\theta_{t+1})] \leq \mathbb{E}[f(\theta_t)] - \frac{\eta}{2}\mathbb{E}[\|\nabla f(\theta_t)\|^2] - \frac{\eta}{2}(1-L\eta)\mathbb{E}[\|v_t\|^2] + \frac{\eta}{2}\mathbb{E}[\|v_t - \nabla f(\theta_t)\|^2].
\end{equation}

\textbf{Step 2: Substitute the MSE bound.}
By Lemma~\ref{lem:mse}, substituting the decomposed tracking error and bias:
\begin{equation}
\label{eq:step2}
\mathbb{E}[f(\theta_{t+1})] \leq \mathbb{E}[f(\theta_t)] - \frac{\eta}{2}\mathbb{E}[\|\nabla f(\theta_t)\|^2] - \frac{\eta}{2}(1-L\eta)\mathbb{E}[\|v_t\|^2] + \eta\mathbb{E}[\|v_t - \bar{g}(\theta_t)\|^2] + \frac{\eta}{4} L_g^2 \nu^2.
\end{equation}

\textbf{Step 3: Telescope over one epoch.}
Summing~\eqref{eq:step2} from $t = 1$ to $t = m$ within epoch $s$:
\begin{align}
\frac{\eta}{2}\sum_{t=1}^{m}\mathbb{E}\left[\left\|\nabla f(\theta_t^{(s)})\right\|^2\right] &\leq \mathbb{E}[f(\theta_1^{(s)})] - \mathbb{E}[f(\theta_{m+1}^{(s)})] - \frac{\eta}{2}(1-L\eta)\sum_{t=1}^{m}\mathbb{E}[\|v_t\|^2] \nonumber\\
&\quad + \eta\sum_{t=1}^{m}\mathbb{E}[\|v_t - \bar{g}(\theta_t)\|^2] + \frac{m\eta}{4} L_g^2 \nu^2. \label{eq:step3}
\end{align}

\textbf{Step 4: Substitute the cumulative variance bound.}
By Lemma~\ref{lem:cumulative}:
\[
\eta\sum_{t=1}^{m}\mathbb{E}[\|v_t - \bar{g}(\theta_t)\|^2] \leq \eta m\mathbb{E}\left[\left\|v_0^{(s)} - \bar{g}(\theta_0^{(s)})\right\|^2\right] + L_g^2 \eta^3 m\sum_{t=1}^{m}\mathbb{E}[\|v_{t-1}\|^2].
\]

Substituting this back into~\eqref{eq:step3}:
\begin{align}
\frac{\eta}{2}\sum_{t=1}^{m}\mathbb{E}\left[\left\|\nabla f(\theta_t^{(s)})\right\|^2\right] &\leq [f(\theta_0^{(s)}) - f^*] + \eta m\mathbb{E}\left[\left\|v_0^{(s)} - \bar{g}(\theta_0^{(s)})\right\|^2\right] + \frac{m\eta}{4} L_g^2 \nu^2 \nonumber\\
&\quad - \left[\frac{\eta}{2}(1-L\eta) - L_g^2 \eta^3 m\right]\sum_{t=1}^{m}\mathbb{E}[\|v_t\|^2]. \label{eq:step4}
\end{align}

\textbf{Step 5: Choose step size to ensure non-negativity.}
Under the conditions $\eta \leq \frac{1}{2L}$ (ensuring $1 - L\eta \geq \frac{1}{2}$) and $\eta \leq \frac{1}{2 L_g \sqrt{m}}$, the composite coefficient satisfies $\left[\frac{\eta}{2}(1-L\eta) - L_g^2 \eta^3 m\right] \ge \frac{\eta}{4} - \frac{\eta}{4} = 0$. This allows us to drop the non-positive $\|v_t\|^2$ sum:
\begin{equation}
\label{eq:step5}
\frac{\eta}{2}\sum_{t=1}^{m}\mathbb{E}\left[\left\|\nabla f(\theta_t^{(s)})\right\|^2\right] \leq [f(\theta_0^{(s)}) - f^*] + \eta m\mathbb{E}\left[\left\|v_0^{(s)} - \bar{g}(\theta_0^{(s)})\right\|^2\right] + \frac{m\eta}{4} L_g^2 \nu^2.
\end{equation}

\textbf{Step 6: Sum over epochs and divide.}
Summing~\eqref{eq:step5} over $s = 0, 1, \ldots, S-1$ and using the telescoping property $\sum_{s=0}^{S-1}[f(\theta_0^{(s)}) - \mathbb{E} f(\theta_{m+1}^{(s)})] \leq f(\theta_0) - f^*$:
\[
\frac{\eta}{2}\sum_{s=0}^{S-1}\sum_{t=1}^{m}\mathbb{E}\left[\left\|\nabla f(\theta_t^{(s)})\right\|^2\right] \leq [f(\theta_0) - f^*] + \eta m S\bar{\sigma}_0^2 + \frac{S m\eta}{4} L_g^2 \nu^2.
\]
Dividing both sides by $\frac{\eta S m}{2}$ yields the following upper bound:
\[
\frac{1}{Sm}\sum_{s=0}^{S-1}\sum_{t=1}^{m}\mathbb{E}\left[\left\|\nabla f(\theta_t^{(s)})\right\|^2\right] \leq \frac{2[f(\theta_0) - f^*]}{\eta m S} + 2\bar{\sigma}_0^2 + \frac{L_g^2}{2} \nu^2.
\]

This completes the proof. \hfill \qedsymbol

\newpage 
\section{Additional Experimental Details}\label{sec:appendix-additional-exp}

This appendix provides granular tabular data supporting the empirical observations discussed in Section~\ref{sec_experiment}. It includes exact numerical endpoints for selected trajectories, exhaustive resource exposure summaries, and broader contextual comparisons against external exact-statevector baselines.

\subsection{Trajectory Numerical Endpoints}
Table~\ref{tab:small-maxcut} details the final logged parameters, sampled energies, and variance measurements for the selected 10- and 15-node MaxCut trajectories featured in Figures~\ref{fig:efficient-su2-diagnostics} and~\ref{fig:qaoa-diagnostics}. Note that unequal logged endpoints preclude a strict matched-budget interpretation.

\begin{table*}[h]
  \centering
  \caption{Single-seed MaxCut trajectory summaries on 30\%-density Erd\H{o}s--R\'enyi graphs. ``Variance'' is the variance of sampled output energies, not gradient-estimator variance. Metrics use the final stored parameters except best energy, which is the trajectory minimum. Unequal logged endpoints preclude a strict matched-budget interpretation.}
  \label{tab:small-maxcut}
  \footnotesize
  \setlength{\tabcolsep}{3.5pt}
  \begin{tabular}{@{}llrrrrr@{}}
    \toprule
    Ansatz / nodes & Method & Best energy & Final energy & Variance & Entropy & Logged evals \\
    \midrule
    EfficientSU2 / 10 & SPSA & $-5.09$ & $-5.03$ & $5.19$ & $4.34$ & 4,982 \\
     & non-recursive CRVG & $-7.93$ & $-7.93$ & $0.26$ & $0.24$ & 5,776 \\
     & recursive CRVG & $-1.97$ & $-0.60$ & $16.68$ & $7.06$ & 4,638 \\
    \addlinespace
    EfficientSU2 / 15 & SPSA & $-7.09$ & $-6.94$ & $21.62$ & $6.34$ & 4,982 \\
     & non-recursive CRVG & $-9.34$ & $-9.34$ & $7.39$ & $4.28$ & 5,951 \\
     & recursive CRVG & $-0.93$ & $0.23$ & $17.18$ & $7.93$ & 4,826 \\
    \addlinespace
    QAOA / 10 & SPSA & $-2.55$ & $-2.11$ & $12.11$ & $7.17$ & 4,982 \\
     & non-recursive CRVG & $-5.87$ & $-5.73$ & $8.13$ & $6.40$ & 5,934 \\
     & recursive CRVG & $-5.73$ & $-5.55$ & $6.98$ & $6.58$ & 4,981 \\
    \addlinespace
    QAOA / 15 & SPSA & $-9.11$ & $-8.91$ & $15.46$ & $7.77$ & 2,762 \\
     & non-recursive CRVG & $-9.16$ & $-8.25$ & $20.94$ & $7.77$ & 5,934 \\
     & recursive CRVG & $-10.29$ & $-9.95$ & $15.25$ & $7.56$ & 4,805 \\
    \bottomrule
  \end{tabular}
\end{table*}

\subsection{Exhaustive Resource Summaries}
Tables~\ref{tab:qaoa-resource-summary} and \ref{tab:mis-qaoa-resource-summary} provide the exact median approximations, sampled-objective call counts, and relative time ratios for the MaxCut QAOA and MIS QAOA depth studies, respectively. Calls are reconstructed from the optimizer control flow, with each invoking one 256-shot Aer sampler job.

\begin{table*}[h]
  \centering
  \caption{Endpoint resource statistics for the paired local QAOA runs. ``Calls'' are calls to the sampled objective, reconstructed from the optimizer control flow; each invokes one 256-shot Aer sampler job. They are not the number of stored history points. Ratios are computed within each instance before taking the median; values below one favor the listed method. Wall time is specific to the archived Aer implementation and execution host.}
  \label{tab:qaoa-resource-summary}
  \footnotesize
  \setlength{\tabcolsep}{3.5pt}
  \begin{tabular}{@{}rllrrrrr@{}}
    \toprule
    Depth & Method & Best ratio & Calls & Time (s) & Calls/SPSA & Time/SPSA & Early stops \\
    \midrule
    1 & SPSA & $0.793$ & 3,247 & 19.3 & $1.00$ & $1.00$ & 36/36 \\
     & non-recursive CRVG & $0.793$ & 3,536 & 20.5 & $1.09$ & $1.06$ & 36/36 \\
     & recursive CRVG & $0.791$ & 3,000 & 17.5 & $0.90$ & $0.87$ & 36/36 \\
    \addlinespace
    3 & SPSA & $0.883$ & 4,500 & 43.8 & $1.00$ & $1.00$ & 36/36 \\
     & non-recursive CRVG & $0.881$ & 3,552 & 37.0 & $0.79$ & $0.80$ & 36/36 \\
     & recursive CRVG & $0.879$ & 2,992 & 31.4 & $0.66$ & $0.68$ & 36/36 \\
    \addlinespace
    5 & SPSA & $0.873$ & 4,500 & 73.6 & $1.00$ & $1.00$ & 36/36 \\
     & non-recursive CRVG & $0.850$ & 3,538 & 59.1 & $0.79$ & $0.80$ & 36/36 \\
     & recursive CRVG & $0.856$ & 2,958 & 49.1 & $0.66$ & $0.66$ & 36/36 \\
    \addlinespace
    7 & SPSA & $0.843$ & 4,500 & 88.3 & $1.00$ & $1.00$ & 36/36 \\
     & non-recursive CRVG & $0.863$ & 3,570 & 73.0 & $0.79$ & $0.79$ & 36/36 \\
     & recursive CRVG & $0.763$ & 2,952 & 60.5 & $0.66$ & $0.66$ & 36/36 \\
    \bottomrule
  \end{tabular}
\end{table*}

\begin{table*}[h]
  \centering
  \caption{MIS QAOA endpoint feasibility, feasible-set sizes, and resources. Entries are descriptive medians across instances; size entries show median (valid support) because size is undefined when no feasible shot was observed. Call and time ratios are paired within instance against SPSA before taking the median; values below one favor the listed method. Calls are sampled-objective calls with 256 shots each, and wall time is host-specific.}
  \label{tab:mis-qaoa-resource-summary}
  \scriptsize
  \setlength{\tabcolsep}{2pt}
  \begin{tabular}{@{}rlrrrrrr@{}}
    \toprule
    Depth & Method & Feasible (\%) & Avg. $|IS|$ (support) & Max. $|IS|$ (support) & Calls & Calls/SPSA & Time/SPSA \\
    \midrule
    1 & SPSA & 49.8 & 2.83 (35/36) & 5.00 (35/36) & 4,123 & $1.00$ & $1.00$ \\
     & non-recursive CRVG & 50.8 & 2.52 (34/36) & 5.00 (34/36) & 3,536 & $0.86$ & $0.86$ \\
     & recursive CRVG & 50.0 & 2.58 (33/36) & 5.00 (33/36) & 3,000 & $0.73$ & $0.73$ \\
    \addlinespace
    3 & SPSA & 68.0 & 2.26 (36/36) & 5.00 (36/36) & 4,500 & $1.00$ & $1.00$ \\
     & non-recursive CRVG & 73.2 & 3.07 (36/36) & 5.50 (36/36) & 3,552 & $0.79$ & $0.81$ \\
     & recursive CRVG & 69.9 & 2.95 (36/36) & 5.00 (36/36) & 2,992 & $0.66$ & $0.68$ \\
    \addlinespace
    5 & SPSA & 73.6 & 2.06 (36/36) & 4.50 (36/36) & 4,500 & $1.00$ & $1.00$ \\
     & non-recursive CRVG & 56.6 & 2.98 (35/36) & 5.00 (35/36) & 3,538 & $0.79$ & $0.80$ \\
     & recursive CRVG & 39.3 & 2.58 (35/36) & 4.00 (35/36) & 2,958 & $0.66$ & $0.67$ \\
    \addlinespace
    7 & SPSA & 68.6 & 2.29 (36/36) & 5.00 (36/36) & 4,500 & $1.00$ & $1.00$ \\
     & non-recursive CRVG & 52.3 & 2.79 (34/36) & 4.50 (34/36) & 3,570 & $0.79$ & $0.80$ \\
     & recursive CRVG & 11.5 & 2.40 (34/36) & 4.00 (34/36) & 2,952 & $0.66$ & $0.66$ \\
    \bottomrule
  \end{tabular}
\end{table*}

\subsection{External Exact-Statevector Baselines}
Tables~\ref{tab:qaoa-external-baselines} and~\ref{tab:qaoa-external-baselines-continued} contextualize the local stochastic methods by reporting every available external exact-statevector parameter-setting baseline for MaxCut QAOA. Because these methods feature unequal instance coverage and unmatched initialization, tuning, and stopping conditions, they are included strictly as descriptive, contextual baselines rather than a controlled ranking.

\begin{table*}[h]
  \centering
  \caption{External exact-statevector QAOA parameter-setting baselines on their available instance support. Baseline entries report median [interquartile range]. For every row, the local median and paired difference are recomputed using only the instances available for that baseline; positive differences favor the best local finite-shot optimizer. Different rows have different support and do not constitute a common controlled ranking.}
  \label{tab:qaoa-external-baselines}
  \footnotesize
  \setlength{\tabcolsep}{4pt}
  \begin{tabular}{@{}rlrrrr@{}}
    \toprule
    Depth & External method & Coverage & Baseline ratio [IQR] & Matched local & Local $-$ baseline \\
    \midrule
    1 & Fixed angle & 15/36 & $0.787$ [$0.773$, $0.817$] & $0.806$ & $+0.020$ \\
     & Fixed angle + opt. & 19/36 & $0.788$ [$0.775$, $0.831$] & $0.806$ & $+0.014$ \\
     & Recursive transfer + opt. & 2/36$^{\dagger}$ & $0.740$ [$0.723$, $0.757$] & $0.753$ & $+0.014$ \\
     & TQA & 27/36 & $0.773$ [$0.753$, $0.788$] & $0.798$ & $+0.025$ \\
     & TQA + opt. & 21/36 & $0.774$ [$0.745$, $0.784$] & $0.790$ & $+0.015$ \\
    \addlinespace
    3 & Fixed angle & 27/36 & $0.892$ [$0.873$, $0.899$] & $0.884$ & $-0.002$ \\
     & Fixed angle + opt. & 29/36 & $0.900$ [$0.889$, $0.914$] & $0.884$ & $-0.011$ \\
     & Fourier & 27/36 & $0.860$ [$0.808$, $0.904$] & $0.890$ & $+0.046$ \\
     & Interpolation & 32/36 & $0.891$ [$0.866$, $0.905$] & $0.885$ & $-0.004$ \\
     & Recursive transfer + angle opt. & 2/36$^{\dagger}$ & $0.863$ [$0.850$, $0.876$] & $0.847$ & $-0.016$ \\
     & Recursive transfer + opt. & 2/36$^{\dagger}$ & $0.844$ [$0.829$, $0.859$] & $0.847$ & $+0.003$ \\
     & TQA & 31/36 & $0.822$ [$0.786$, $0.835$] & $0.884$ & $+0.069$ \\
     & TQA + opt. & 25/36 & $0.874$ [$0.862$, $0.892$] & $0.884$ & $-0.003$ \\
     & Transfer strategy & 32/36 & $0.881$ [$0.858$, $0.890$] & $0.884$ & $+0.011$ \\
    \bottomrule
  \end{tabular}
  \par\vspace{2pt}\scriptsize $^{\dagger}$Sparse support ($n<5$); interpret the median descriptively only.
\end{table*}

\begin{table*}[h]
  \centering
  \caption{External exact-statevector QAOA parameter-setting baselines on their available instance support (continued). Baseline entries report median [interquartile range]. For every row, the local median and paired difference are recomputed using only the instances available for that baseline; positive differences favor the best local finite-shot optimizer. Different rows have different support and do not constitute a common controlled ranking.}
  \label{tab:qaoa-external-baselines-continued}
  \footnotesize
  \setlength{\tabcolsep}{4pt}
  \begin{tabular}{@{}rlrrrr@{}}
    \toprule
    Depth & External method & Coverage & Baseline ratio [IQR] & Matched local & Local $-$ baseline \\
    \midrule
    5 & Fixed angle & 13/36 & $0.916$ [$0.903$, $0.940$] & $0.893$ & $-0.024$ \\
     & Fixed angle + opt. & 14/36 & $0.935$ [$0.927$, $0.944$] & $0.893$ & $-0.038$ \\
     & Fourier & 26/36 & $0.828$ [$0.782$, $0.878$] & $0.893$ & $+0.050$ \\
     & Interpolation & 31/36 & $0.932$ [$0.912$, $0.949$] & $0.891$ & $-0.041$ \\
     & Recursive transfer + angle opt. & 2/36$^{\dagger}$ & $0.917$ [$0.909$, $0.924$] & $0.799$ & $-0.117$ \\
     & Recursive transfer + opt. & 2/36$^{\dagger}$ & $0.882$ [$0.869$, $0.895$] & $0.799$ & $-0.083$ \\
     & TQA & 30/36 & $0.876$ [$0.846$, $0.893$] & $0.889$ & $+0.021$ \\
     & TQA + opt. & 24/36 & $0.918$ [$0.899$, $0.929$] & $0.883$ & $-0.034$ \\
     & Transfer strategy & 28/36 & $0.909$ [$0.894$, $0.932$] & $0.885$ & $-0.030$ \\
    \addlinespace
    7 & Fixed angle & 6/36 & $0.949$ [$0.927$, $0.960$] & $0.898$ & $-0.049$ \\
     & Fixed angle + opt. & 6/36 & $0.963$ [$0.958$, $0.965$] & $0.898$ & $-0.065$ \\
     & Fourier & 26/36 & $0.862$ [$0.788$, $0.908$] & $0.894$ & $+0.028$ \\
     & Interpolation & 31/36 & $0.954$ [$0.940$, $0.966$] & $0.890$ & $-0.067$ \\
     & Recursive transfer + angle opt. & 2/36$^{\dagger}$ & $0.941$ [$0.938$, $0.945$] & $0.800$ & $-0.141$ \\
     & Recursive transfer + opt. & 2/36$^{\dagger}$ & $0.906$ [$0.894$, $0.917$] & $0.800$ & $-0.106$ \\
     & TQA & 31/36 & $0.906$ [$0.876$, $0.920$] & $0.890$ & $-0.012$ \\
     & TQA + opt. & 24/36 & $0.941$ [$0.928$, $0.946$] & $0.883$ & $-0.060$ \\
     & Transfer strategy & 31/36 & $0.935$ [$0.916$, $0.954$] & $0.890$ & $-0.050$ \\
    \bottomrule
  \end{tabular}
  \par\vspace{2pt}\scriptsize $^{\dagger}$Sparse support ($n<5$); interpret the median descriptively only.
\end{table*}

\newpage
\bibliographystyle{plainnat}
\bibliography{reference}

@article{cerezo2021variational,
  title={Variational quantum algorithms},
  author={Cerezo, Marco and Arrasmith, Andrew and Babbush, Ryan and Benjamin, Simon C and Endo, Suguru and Fujii, Keisuke and McClean, Jarrod R and Mitarai, Kosuke and Yuan, Xiao and Cincio, Lukasz and others},
  journal={Nature Reviews Physics},
  volume={3},
  number={9},
  pages={625--644},
  year={2021},
  publisher={Nature Publishing Group UK London}
}

@article{mitarai2018quantum,
  title={Quantum circuit learning},
  author={Mitarai, Kosuke and Negoro, Makoto and Kitagawa, Masahiro and Fujii, Keisuke},
  journal={Physical Review A},
  volume={98},
  number={3},
  pages={032309},
  year={2018},
  publisher={APS}
}

@article{sweke2020stochastic,
  title={Stochastic gradient descent for hybrid quantum-classical optimization},
  author={Sweke, Ryan and Wilde, Frederik and Meyer, Johannes and Schuld, Maria and F{\"a}hrmann, Paul K and Meynard-Piganeau, Barth{\'e}l{\'e}my and Eisert, Jens},
  journal={Quantum},
  volume={4},
  pages={314},
  year={2020},
  publisher={Verein zur F{\"o}rderung des Open Access Publizierens in den Quantenwissenschaften}
}

@article{arrasmith2020operator,
  title={Operator sampling for shot-frugal optimization in variational algorithms},
  author={Arrasmith, Andrew and Cincio, Lukasz and Somma, Rolando D and Coles, Patrick J},
  journal={arXiv preprint arXiv:2004.06252},
  year={2020}
}

@article{gacon2021simultaneous,
  title={Simultaneous perturbation stochastic approximation of the quantum fisher information},
  author={Gacon, Julien and Zoufal, Christa and Carleo, Giuseppe and Woerner, Stefan},
  journal={Quantum},
  volume={5},
  pages={567},
  year={2021},
  publisher={Verein zur F{\"o}rderung des Open Access Publizierens in den Quantenwissenschaften}
}

@article{spall2002multivariate,
  title={Multivariate stochastic approximation using a simultaneous perturbation gradient approximation},
  author={Spall, James C},
  journal={IEEE transactions on automatic control},
  volume={37},
  number={3},
  pages={332--341},
  year={2002},
  publisher={IEEE}
}

@article{johnson2013accelerating,
  title={Accelerating stochastic gradient descent using predictive variance reduction},
  author={Johnson, Rie and Zhang, Tong},
  journal={Advances in neural information processing systems},
  volume={26},
  year={2013}
}

@article{liu2018zeroth,
  title={Zeroth-order stochastic variance reduction for nonconvex optimization},
  author={Liu, Sijia and Kailkhura, Bhavya and Chen, Pin-Yu and Ting, Paishun and Chang, Shiyu and Amini, Lisa},
  journal={Advances in neural information processing systems},
  volume={31},
  year={2018}
}

@inproceedings{nguyen2017sarah,
  title={SARAH: A novel method for machine learning problems using stochastic recursive gradient},
  author={Nguyen, Lam M and Liu, Jie and Scheinberg, Katya and Tak{\'a}{\v{c}}, Martin},
  booktitle={International conference on machine learning},
  pages={2613--2621},
  year={2017},
  organization={PMLR}
}

@article{sidford2023quantum,
  title={Quantum speedups for stochastic optimization},
  author={Sidford, Aaron and Zhang, Chenyi},
  journal={Advances in Neural Information Processing Systems},
  volume={36},
  pages={35300--35330},
  year={2023}
}

@article{schuld2019evaluating,
  title={Evaluating analytic gradients on quantum hardware},
  author={Schuld, Maria and Bergholm, Ville and Gogolin, Christian and Izaac, Josh and Killoran, Nathan},
  journal={Physical Review A},
  volume={99},
  number={3},
  pages={032331},
  year={2019},
  publisher={APS}
}

@article{mari2021estimating,
  title={Estimating the gradient and higher-order derivatives on quantum hardware},
  author={Mari, Andrea and Bromley, Thomas R and Killoran, Nathan},
  journal={Physical Review A},
  volume={103},
  number={1},
  pages={012405},
  year={2021},
  publisher={APS}
}

@article{larocca2025barren,
  title={Barren plateaus in variational quantum computing},
  author={Larocca, Martin and Thanasilp, Supanut and Wang, Samson and Sharma, Kunal and Biamonte, Jacob and Coles, Patrick J and Cincio, Lukasz and McClean, Jarrod R and Holmes, Zo{\"e} and Cerezo, Marco},
  journal={Nature Reviews Physics},
  volume={7},
  number={4},
  pages={174--189},
  year={2025},
  publisher={Nature Publishing Group UK London}
}

@article{preskill2018quantum,
  title={Quantum computing in the NISQ era and beyond},
  author={Preskill, John},
  journal={Quantum},
  volume={2},
  pages={79},
  year={2018},
  publisher={Verein zur F{\"o}rderung des Open Access Publizierens in den Quantenwissenschaften}
}

@article{peruzzo2014variational,
  title={A variational eigenvalue solver on a photonic quantum processor},
  author={Peruzzo, Alberto and McClean, Jarrod and Shadbolt, Peter and Yung, Man-Hong and Zhou, Xiao-Qi and Love, Peter J and Aspuru-Guzik, Al{\'a}n and O’brien, Jeremy L},
  journal={Nature communications},
  volume={5},
  number={1},
  pages={4213},
  year={2014},
  publisher={Nature Publishing Group UK London}
}

@article{farhi2014quantum,
  title={A quantum approximate optimization algorithm},
  author={Farhi, Edward and Goldstone, Jeffrey and Gutmann, Sam},
  journal={arXiv preprint arXiv:1411.4028},
  year={2014}
}

@article{kandala2017hardware,
  title={Hardware-efficient variational quantum eigensolver for small molecules and quantum magnets},
  author={Kandala, Abhinav and Mezzacapo, Antonio and Temme, Kristan and Takita, Maika and Brink, Markus and Chow, Jerry M and Gambetta, Jay M},
  journal={nature},
  volume={549},
  number={7671},
  pages={242--246},
  year={2017},
  publisher={Nature Publishing Group UK London}
}

@article{harrigan2021quantum,
  title={Quantum approximate optimization of non-planar graph problems on a planar superconducting processor},
  author={Harrigan, Matthew P and Sung, Kevin J and Neeley, Matthew and Satzinger, Kevin J and Arute, Frank and Arya, Kunal and Atalaya, Juan and Bardin, Joseph C and Barends, Rami and Boixo, Sergio and others},
  journal={Nature Physics},
  volume={17},
  number={3},
  pages={332--336},
  year={2021},
  publisher={Nature Publishing Group UK London}
}

@article{kubler2020adaptive,
  title={An adaptive optimizer for measurement-frugal variational algorithms},
  author={K{\"u}bler, Jonas M and Arrasmith, Andrew and Cincio, Lukasz and Coles, Patrick J},
  journal={Quantum},
  volume={4},
  pages={263},
  year={2020},
}

@article{nakanishi2020sequential,
  title={Sequential minimal optimization for quantum-classical hybrid algorithms},
  author={Nakanishi, Ken M and Fujii, Keisuke and Todo, Synge},
  journal={Physical Review Research},
  volume={2},
  number={4},
  pages={043158},
  year={2020},
}

@article{stokes2020quantum,
  title={Quantum natural gradient},
  author={Stokes, James and Izaac, Josh and Killoran, Nathan and Carleo, Giuseppe},
  journal={Quantum},
  volume={4},
  pages={269},
  year={2020},
}

@article{ghadimi2013stochastic,
  title={Stochastic first-and zeroth-order methods for nonconvex stochastic programming},
  author={Ghadimi, Saeed and Lan, Guanghui},
  journal={SIAM journal on optimization},
  volume={23},
  number={4},
  pages={2341--2368},
  year={2013},
  publisher={SIAM}
}

@article{sack2024large,
  title={Large-scale quantum approximate optimization on nonplanar graphs with machine learning noise mitigation},
  author={Sack, Stefan H and Egger, Daniel J},
  journal={Physical Review Research},
  volume={6},
  number={1},
  pages={013223},
  year={2024},
  publisher={APS}
}





\end{document}